\newif\ifoup \ouptrue
\oupfalse

\ifoup
\documentclass[webpdf,modern,large]{oup-authoring-template}%
\else
\documentclass[a4paper,DIV=15,10pt]{scrartcl}
\fi

\usepackage{amsmath,amssymb,amsthm}
\usepackage{mathrsfs}
\usepackage{mleftright}
\usepackage{textcomp}
\usepackage{url}
\usepackage{hyperref}
\usepackage{bussproofs, bussproofs-extra}
\usepackage[l2tabu, orthodox]{nag}
\usepackage{tikz}
\usetikzlibrary{positioning}
\usetikzlibrary{tikzmark,calc,topaths}

\usepackage{mysettings}
\newcommand{\mathdom}{\mathrm{dom}}

\newcommand{\mathnat}{\mathbb{N}}
\newcommand{\mathnatpos}{\mathbb{N}_{>0}}
\newcommand{\mathnatKC}{\mathbb{N}^{*}}
\newcommand{\BNFeq}{\mathrel{::=}}

\newcommand{\LKIDOm}{\ensuremath{\mathtt{LKID}^{\omega}}}
\newcommand{\CLKID}{\ensuremath{\mathtt{CLKID}^{\omega}}}
\newcommand{\CLKIDa}{\ensuremath{\mathtt{CLKID}^{\omega}_a}}

\newcommand{\mathrelbar}{\mathrel{|}}

\newcommand{\mathsetintension}[2]{%
\mleft\{ #1 \ \mleft| \  #2 \mright. \mright\}%
}
\newcommand{\mathsetextension}[1]{%
\mleft\{ #1 \mright\}%
}
\newcommand{\mathvect}[1]{\mathbf{#1}}
\newcommand{\mathof}[2]{\mathop{#1}\mleft(#2\mright)}

\newcommand{\mathdefparfunc}[3]{#1\colon #2\rightharpoonup #3}

\newcommand{\mathFV}[1]{\mathop{\mathrm{FV}}\mleft( #1 \mright)}

\newcommand{\mathSeq}[2]{\mleft( \mathof{#1}{#2} \mright)_{0}}
\newcommand{\mathRule}[2]{\mleft( \mathof{#1}{#2} \mright)_{1}}

\newcommand{\mathRuleSet}{\mathrm{Rule}}
\newcommand{\mathSeqSet}{\mathrm{Seq}}

\newcommand{\mathBudset}[1]{\mathop{\mathrm{Bud}}\mleft({#1}\mright)}
\newcommand{\mathcompanion}[1]{\mathop{\mathcal{#1}}}

\newcommand{\mathtreeunfolding}[1]{\mathop{\mathrm{T}}\mleft(#1\mright)}

\newcommand{\mathdertree}[1]{\mathcal{#1}}
\newcommand{\mathprooffig}[1]{\mathcal{#1}}

\newcommand{\mathindAddonesy}{\mathrm{Add_1}}
\newcommand{\mathindAddtwosy}{\mathrm{Add_2}}
\newcommand{\mathindAddone}[3]{\mathop{\mathindAddonesy}\mleft(#1, #2, #3\mright)}
\newcommand{\mathindAddtwo}[3]{\mathop{\mathindAddtwosy}\mleft(#1, #2, #3\mright)}

\newcommand{\mathvars}[1]{\mathop{\mathrm{VC}}\mleft(#1\mright)}
\newcommand{\mathsetA}[1]{\mathop{\mathrm{A}}\mleft(#1\mright)}
\newcommand{\mathsetBC}[1]{\mathop{\mathrm{BC}}\mleft(#1\mright)}
\newcommand{\mathsetBone}[1]{\mathop{\mathrm{B_{1}}}\mleft(#1\mright)}
\newcommand{\mathsetC}[1]{\mathop{\mathrm{C}}\mleft(#1\mright)}

\newcommand{\mathequivrel}[1]{\mathrel{\cong_{#1}}}
\newcommand{\mathdeprel}[1]{\mathrel{\sim_{#1}}}

\newcommand{\rulenameLa}{\rulename{\ensuremath{=} \ensuremath{{\text{L}}_{\text{a}}}}}

\newcommand{\mathdertreecf}{\mathdertree{D}_{\mathrm{cf}}}
\newcommand{\mathcompanioncf}{\mathcompanion{C}_{\mathrm{cf}}}
\newcommand{\mathproofcf}{\mleft(\mathdertreecf, \mathcompanioncf\mright)}

\newcounter{outlinelemmacount}
\newcommand{\outlinelemma}{\refstepcounter{outlinelemmacount}(\theoutlinelemmacount)}

\makeatletter
\renewcommand\p@enumii{}
\makeatother
\onecolumn

\begin{document}

\ifoup \renewenvironment*{proof}[1][\proofname]{%
\par \medskip
\noindent \textbf{#1. } 
}%
{\qed \par \medskip}%

 \else
\sloppy 
\hbadness=10000
\fi

\ifoup %#!latexmk -c -gg -lualatex main.tex
\title[Counterexample to cut-elimination in cyclic proof system]{Counterexample to cut-elimination in cyclic proof system for first-order logic with inductive definitions}

\author[1,$\ast$]{Yukihiro Masuoka}
\author[1]{Makoto Tatsuta}

\address[1]{\orgdiv{Department of Informatics}, \orgname{The Graduate University for Advanced Studies}, \orgaddress{\street{Chiyoda-ku}, \postcode{101-8430}, \state{Tokyo}, \country{Japan}}}

\corresp[$\ast$]{Corresponding author. \href{email:yukihiro\_m@nii.ac.jp}{yukihiro\_m@nii.ac.jp}}

% oup ------------------------------------------------------------------
\journaltitle{Journal of Logic and Computation}
\DOI{ }
\copyrightyear{2020}
\pubyear{ }
\access{ }
\appnotes{ }

\firstpage{1}
% ----------------------------------------------------------------------

\abstract{
 A cyclic proof system is a proof system
 whose proof figure is a tree with cycles.
 The cut-elimination in a proof system is fundamental.
 It is conjectured that
 the cut-elimination in the cyclic proof system for first-order logic with inductive definitions
 does not hold.
 This paper shows that the conjecture is correct by giving a sequent not provable without the cut rule
 but provable in the cyclic proof system.
 }


\maketitle
 \else %#!latexmk -c -gg -lualatex main.tex
\title{The failure of cut-elimination in cyclic proof for first-order logic with inductive definitions}

\author{Yukihiro Oda\thanks{Department of Informatics, The Graduate University for Advanced Studies, yukihiro\_m [at] nii.ac.jp} \and James Brotherston\thanks{Department of Computer Science, University College London} \and Makoto Tatsuta\thanks{National Institute of Informatics / Sokendai}}

\date{\today ver.}

\maketitle

\begin{abstract}
 A cyclic proof system is a proof system
 whose proof figure is a tree with cycles.
 The cut-elimination in a proof system is fundamental.
 It is conjectured that
 the cut-elimination in the cyclic proof system for first-order logic with inductive definitions
 does not hold.
 This paper shows that the conjecture is correct by giving a sequent not provable without the cut rule
 but provable in the cyclic proof system.
\end{abstract}
 \fi
%#!latexmk -c -gg -lualatex main.tex
 \section{Introduction}
 % Definition of the topic plus background (1-3)
 A cyclic proof system, or a circular proof system, 
 is a proof system whose proof figure is a tree with cycles \cite{BrotherstonPhD}. 
 Such proof systems have been used to formalize several logics and theories,
 such as modal $\mu$-calculus 
 \cite{Stirling1991, Stirling2014, Afshari2017}, 
 linear time $\mu$-calculus 
 \cite{Brunnler2008, Doumane2016,Kokkinis2016},
 linear logic with fixed points 
 \cite{Doumane2017, Nollet2018},
 G\"{o}del-L\"{o}b provability logic 
 \cite{Shamkanov2014},
 first-order $\mu$-calculus \cite{Sprenger2003},
 first-order logic with inductive definitions
 \cite{Brotherston2005, Brotherston2011, Berardi2017c},
 arithmetic \cite{Simpson2017, Berardi2017e},
 bunched logic \cite{Brotherston2007},
 separation logic 
 \cite{Brotherston2008, Kimura2020, Saotome2020, Tatsuta2019},
 and Kleene algebra \cite{Das2017}.
 Cyclic proofs are also useful
 for software verification, including 
 verifying properties of concurrent processes 
 \cite{Schopp2002},
 termination of pointer programs 
 \cite{Brotherston2008}, and
 decision procedures for symbolic heaps 
 \cite{Brotherston2012, Chu2015, Ta2017, Tatsuta2019}.
 
 % accepted state of the art plus problem to be resolved (2-4)*
 The cut-elimination property is
 a fundamental property of a proof system.
 For example, the cut-elimination theorem 
 for first-order logic immediately implies 
 consistency, 
 the subformula property and 
 Craig's interpolation theorem 
 (see \cite{Buss1998i}).
 
 % authors' objectives (1-2)*
 % aim of the present work (1-2)*
 Despite its importance, 
 it was an open problem
 whether the cut-elimination property
 in the cyclic proof system for
 first-order logic with inductive definitions
 holds. 
 In Conjecture 5.2.4.\ of \cite{BrotherstonPhD}, 
 Brotherston has conjectured that 
 the cut-elimination property
 in the system does not hold. 
 
 This paper presents a counterexample to
 cut-elimination in the cyclic proof system for
 first-order logic with inductive definitions.
 In other words, we show that
 the conjecture is correct.
 Our counterexample is a sequent that says
 an addition predicate implies another addition predicate
 with a different definition.
 In order to show it is not cut-free provable,
 under the assumption that it is cut-free provable,
 we construct an infinite sequence of sequents in a finite cyclic proof figure,
 which leads to a contradiction.
 For this purpose, we use \emph{cycle-normalization} \cite{BrotherstonPhD}, and
 we also give a simpler proof of it.

 % introduction to the literature
 % survey of pertinent literature
 There exist cut-free and complete cyclic proof systems for some logics and theories, 
 including modal $\mu$-calculus \cite{Stirling2014, Afshari2017}, 
 linear time $\mu$-calculus  \cite{Brunnler2008},
 G\"{o}del-L\"{o}b provability logic \cite{Shamkanov2014}
 and Kleene algebra \cite{Das2017}.
 
 Kimura et al.\ \cite{Kimura2020} give
 a counterexample to cut-elimination
 in cyclic proofs for separation logic.
 They also suggest their proof technique 
 cannot be applied to show a counterexample to cut-elimination
 in a cyclic proof system with contraction and 
 weakening on antecedents \cite{Kimura2020}.
 Their proof technique is to give a path in a cut-free proof 
 that contradicts a soundness condition.
 % authors' contribution (1-2)*
 % main result / conclusions (1-4)
 In the counterexample we give in this article, 
 constructing such a path seems complicated
 because the system we consider is a system with contraction and weakening of antecedents.
 % % future implications (1-2)
 
 % outline of structure (3-4 very short sentences)
 Section \ref{sec:syntax} describes the language for first-order logic with inductive definitions.
 Section \ref{sec:CLKID} introduces Brotherston's cyclic proof system \CLKID. 
 Section \ref{sec:main} gives a sequent that is not cut-free provable in \CLKID\ but provable in \CLKID.
 
%#!latexmk -c -gg -lualatex main.tex
 \section{Language}
 \label{sec:syntax}
 In this section, we give  the syntax of a language for first-order logic with inductive definitions. 
 The language is the same as that given in \cite{Brotherston2011}.

\emph{Terms} are defined by 
 \[
 t \BNFeq x \mathrelbar f\overbrace{t\cdots t}^{\text{$n$}},
 \]
 where $x$ is a variable symbol and $f$ is an $n$-ary function symbol.
 We write $\mathvect{x}$ for a sequence of variables and
 $\mathvect{u}\mleft(\mathvect{x}\mright)$ for a sequence of terms 
 in which the variables $\mathvect{x}$ occur.

 \emph{Predicate symbols} consist of \emph{ordinary predicate symbols}, 
 denoted by $Q_{1}, Q_{2}, \ldots$,
 and \emph{inductive predicate symbols}, denoted by $P_{1}, \dots, P_{n}$. 
 Inductive predicate symbols are given with an \emph{inductive definition set}, which we define later.
 We assume that inductive predicate symbols are finite.

 An \emph{atomic formula} is defined as ${t_{1} = t_{2}}$ or  $R(t_{1}, \ldots, t_{n})$
 where $t_{1}, t_{2}, \dots, t_{n}$ are terms and $R$ is an $n$-ary predicate symbol. 
 $Q_{1}\mathvect{u}$ denotes $Q_{1}(\mathvect{u})$, where $\mathvect{u}$ is a sequence of terms.
 \emph{Formulas} are defined by
 \[
 \varphi \BNFeq A \mathrelbar \lnot\varphi \mathrelbar  \varphi \mathrel{\land} \varphi \mathrelbar  \varphi \mathrel{\lor} \varphi \mathrelbar \mathop{\exists x} \varphi \mathrelbar \mathop{\forall x} \varphi,
 \]
 where $A$ is an atomic formula and $x$ is a variable.
 We define \emph{free variables} as usual, and
 $\mathFV{\varphi}$ is defined as the set of free variables in a formula $\varphi$.
 We write
 $\varphi\mleft[x_{0}:=t_{0}, \dots, x_{r}:=t_{r}\mright]$ for a formula obtained from a formula $\varphi$
 by simultaneously substituting terms $t_{0}$, $\ldots$,  $t_{r}$ for variables $x_{0}$, $\ldots$, $x_{r}$,
 respectively.
 We sometimes write $\theta$ for $x_{0}:=t_{0}, \dots, x_{r}:=t_{r}$.

 \begin{definition}[Inductive definition set]
  A \emph{production} is defined as
  \begin{center}
   \begin{inlineprooftree}
    \AxiomC{ $Q_{1} \mathvect{u}_{1} \quad \cdots \quad Q_{h} \mathvect{u}_{h} \quad P_{j_{1}}\mathvect{t}_{1} \quad \cdots \quad P_{j_{m}}\mathvect{t}_{m}$ }
    \UnaryInfC{$P_{i} \mathvect{t}$}
   \end{inlineprooftree},
  \end{center}
  where  $Q_{1}\mathvect{u}_{1}, \dots,  Q_{h}\mathvect{u}_{h}$
  are atomic formulas with ordinal predicate symbols and 
  $P_{j_{1}}\mathvect{t}_{1},  \dots, P_{j_{m}}\mathvect{t}_{m}$ and $P_{i} \mathvect{t}$
  are atomic formulas with inductive predicate symbols.
  
  The formulas above the line of a production are called the \emph{assumptions} of the production. 
  The formula under the line of a production is called the \emph{conclusion} of the production. 
  An \emph{inductive definition set} is a finite set of productions.
 \end{definition}

 \begin{definition}[Sequent]
  A \emph{sequent} is a pair of finite sets of formulas,
  denoted by $\Gamma \fCenter \Delta$, where $\Gamma$, $\Delta$ are finite sets of formulas. 
  $\Gamma$ is called the antecedent of $\Gamma \fCenter \Delta$ and $\Delta$ 
  is called the consequent of $\Gamma \fCenter \Delta$.
 \end{definition}

 For a set of formulas $\Gamma$, we define $\mathFV{\Gamma}$ as the set of 
 free variables of formulas in $\Gamma$. \medskip

 The semantics of inductive predicates is given by
 the least fixed point of a monotone operator constructed from the inductive definition set 
 \cite{Brotherston2011}.
 Since we do not use semantics in this paper, we do not discuss it in detail.

  We write $s^{n}x$ for $\overbrace{s \cdots s}^{n}x$. 
%#!latexmk -c -gg -lualatex main.tex
\section{Cyclic proof system \CLKID\ for first-order logic with inductive definitions}
\label{sec:CLKID}
In this section, we define a cyclic proof system \CLKID\ for first-order logic with inductive definitions.
To define it,  we also define an infinitary proof system \LKIDOm\ with the same language.
Then, \CLKID\ is understood as the subsystem \LKIDOm.
These systems are the same as \CLKID\ and \LKIDOm\ 
defined in \cite{BrotherstonPhD, Brotherston2011}.

  \subsection{Inference rules of \CLKID}
  \label{subsec:inference_rules}

  This section gives the inference rules of \CLKID.
  The inference rules except for rules of inductive predicates 
  are given in Figure \ref{fig:inference-rules}.
  The \emph{principal formula} of a rule is the distinguished formula 
  introduced by the rule in its conclusion.
  We use the commas in sequents for a set union. The contraction rule is implicit. 

  \begin{figure}[tbhp]%
 \centering
 \begin{tabular}[tb]{l}
  {\textbf{Structural rules:}} \smallskip \hfill \\
  
  \begin{tabular}{cc}
   \begin{minipage}{0.5\hsize}
    \begin{prooftree}
     \AxiomC{}
     \RightLabel{(\rulename{Axiom})($\Gamma \cap \Delta \neq \emptyset$)}
     \UnaryInfC{$\Gamma \fCenter \Delta$}
    \end{prooftree}
    \smallskip
   \end{minipage}
   &
       \begin{minipage}{0.5\hsize}
	\begin{prooftree}
	 \AxiomC{$ \Gamma' \fCenter \Delta'$}
	 \RightLabel{(\rulename{Weak})($\Gamma' \subseteq \Gamma$, $\Delta' \subseteq \Delta$)}
	 \UnaryInfC{$ \Gamma \fCenter \Delta $}
	\end{prooftree}
	\smallskip
       \end{minipage}
       \\
   \begin{minipage}{0.5\hsize}
    \begin{prooftree}
     \AxiomC{$\Gamma \fCenter \varphi, \Delta$}
     \AxiomC{$\Gamma, \varphi \fCenter \Delta$}
     \rulenamelabel{Cut}
     \BinaryInfC{$\Gamma \fCenter \Delta$}
    \end{prooftree}
    \smallskip
   \end{minipage}
   &
       \begin{minipage}{0.5\hsize}
	\begin{prooftree}
	 \AxiomC{$\Gamma \fCenter \Delta$}
	 \rulenamelabel{Subst}
	 \UnaryInfC{$\Gamma \left[ \theta \right] \fCenter \Delta \left[ \theta \right] $}
	\end{prooftree}
	\smallskip
       \end{minipage} 
  \end{tabular} \\
  
  \textbf{Logical rules:} \smallskip \hfill \\
  
  \begin{tabular}{cc}             
   \begin{minipage}{0.5\hsize}
    \begin{prooftree}
     \AxiomC{$\Gamma \fCenter \varphi, \Delta $}
     \rulenamelabel{$\lnot$ L}
     \UnaryInfC{$ \Gamma, \lnot\varphi \fCenter \Delta $}
    \end{prooftree}
    \smallskip
   \end{minipage}  
   &
       \begin{minipage}{0.5\hsize}
	\begin{prooftree}
	 \AxiomC{$\Gamma, \varphi \fCenter \Delta $}
	 \rulenamelabel{$\lnot$ R}
	 \UnaryInfC{$ \Gamma  \fCenter \lnot\varphi, \Delta $}
	\end{prooftree}
	\smallskip
       \end{minipage}
       \\
   \begin{minipage}{0.5\hsize}
    \begin{prooftree}
     \AxiomC{$\Gamma, \varphi \fCenter \Delta $}
     \AxiomC{$\Gamma, \psi \fCenter \Delta $}
     \rulenamelabel{$\lor$ L}
     \BinaryInfC{$ \Gamma, \varphi\lor\psi \fCenter \Delta$}
    \end{prooftree}
    \smallskip
   \end{minipage}
   &
       \begin{minipage}{0.5\hsize}
	\begin{prooftree}
	 \AxiomC{$\Gamma \fCenter \varphi, \psi, \Delta $}
	 \rulenamelabel{$\lor$ R}
	 \UnaryInfC{$ \Gamma  \fCenter \varphi\lor \psi, \Delta $}
	\end{prooftree}
       \end{minipage}
       \\
   \begin{minipage}{0.5\hsize}
    \begin{prooftree}
     \AxiomC{$\Gamma, \varphi, \psi \fCenter \Delta $}
     \rulenamelabel{$\land$ L}
     \UnaryInfC{$\Gamma, \varphi \land \psi \fCenter \Delta $}
    \end{prooftree}
   \end{minipage}
   &
       \begin{minipage}{0.5\hsize}
	\begin{prooftree}
	 \AxiomC{$\Gamma \fCenter \varphi, \Delta $}
	 \AxiomC{$\Gamma \fCenter \psi, \Delta $}
	 \rulenamelabel{$\land$ R}
	 \BinaryInfC{$ \Gamma \fCenter \varphi \land \psi, \Delta$}
	\end{prooftree}
	\smallskip
       \end{minipage}
       \\
   \begin{minipage}{0.5\hsize}
    \begin{prooftree}
     \AxiomC{$\Gamma \fCenter \varphi, \Delta $}
     \AxiomC{$\Gamma, \psi \fCenter \Delta $}
     \rulenamelabel{$\to$ L}
     \BinaryInfC{$ \Gamma, \varphi \to \psi \fCenter \Delta$}
    \end{prooftree}
    \smallskip
   \end{minipage}
   &
       \begin{minipage}{0.5\hsize}
	\begin{prooftree}
	 \AxiomC{$\Gamma, \varphi \fCenter \psi, \Delta $}
	 \rulenamelabel{$\to$ R}
	 \UnaryInfC{$ \Gamma  \fCenter \varphi \to \psi, \Delta $}
	\end{prooftree}
       \end{minipage}
       \\
   \begin{minipage}{0.5\hsize}
    \begin{prooftree}
     \AxiomC{$\Gamma, \varphi\left[x := t \right]\fCenter \Delta $}
     \rulenamelabel{$\forall$ L}
     \UnaryInfC{$ \Gamma, \forall x \varphi \fCenter \Delta $}
    \end{prooftree}
    \smallskip
   \end{minipage}  
   &
       \begin{minipage}{0.5\hsize}
	\begin{prooftree}
	 \AxiomC{$\Gamma \fCenter \varphi, \Delta $}
	 \RightLabel{(\rulename{$\forall$ R})($x \not \in \mathFV{ \Gamma \cup \Delta}$)}
	 \UnaryInfC{$ \Gamma  \fCenter \forall x \varphi, \Delta $}
	\end{prooftree}
	\smallskip
       \end{minipage}
       \\
   \begin{minipage}{0.5\hsize}
    \begin{prooftree}
     \AxiomC{$\Gamma, \varphi \fCenter \Delta $}
     \RightLabel{(\rulename{$\exists$ L})($x \not \in \mathFV{ \Gamma \cup \Delta }$)}
     \UnaryInfC{$ \Gamma, \exists x \varphi \fCenter \Delta $}
    \end{prooftree}
    \smallskip
   \end{minipage}  
   &
       \begin{minipage}{0.5\hsize}
	\begin{prooftree}
	 \AxiomC{$\Gamma \fCenter \varphi \left[ x := t \right], \Delta $}
	 \rulenamelabel{$\exists$ R}
	 \UnaryInfC{$ \Gamma  \fCenter \exists x \varphi, \Delta$}
	\end{prooftree}
	\smallskip
       \end{minipage}
  \end{tabular} \\
  \begin{tabular}{cc}
   \begin{minipage}{0.5\hsize}
    \begin{prooftree}
     \AxiomC{$\Gamma \left[x := u, y := t\right] \fCenter \Delta \left[x := u, y := t\right] $}
     % \LeftLabel{$x$, $y\notin \mathVar{t}\cup\mathVar{u}$}
     \rulenamelabel{$=$ L}
     \UnaryInfC{$ \Gamma \left[ x := t, y := u \right], t = u \fCenter \Delta \left[x := t, y := u\right] $}
    \end{prooftree}
    \smallskip
   \end{minipage}  
   &
   \begin{minipage}{0.5\hsize}
    \begin{prooftree}
     \AxiomC{\phantom{$ \Gamma  \fCenter t = t, \Delta $}}
     \rulenamelabel{$=$ R}
     \UnaryInfC{$ \Gamma  \fCenter t = t, \Delta $}
    \end{prooftree}
    \smallskip
   \end{minipage} 
  \end{tabular}
 \end{tabular}
 \caption{Inference rules except rules for inductive predicates} 
 \label{fig:inference-rules}
\end{figure}
  
  We present the two inference rules for inductive predicates.  
  First, for each production 
  \begin{center}
   \begin{inlineprooftree}
    \AxiomC{ $ Q_1 \mathvect{u}_1\mleft( \mathvect{x} \mright) \quad \cdots \quad Q_h \mathvect{u}_h\mleft( \mathvect{x} \mright) \quad P_{j_1} \mathvect{t}_1  \mleft( \mathvect{x} \mright)\quad \cdots \quad P_{j_m} \mathvect{t}_m  \mleft( \mathvect{x} \mright)$ }
    \UnaryInfC{$ P_i \mathvect{t} \mleft( \mathvect{x} \mright)$}
   \end{inlineprooftree},
  \end{center}
  there is the inference rule 
  \begin{center} 
   \begin{inlineprooftree}
    \AxiomC{ $\Gamma \fCenter Q_1 \mathvect{u}_1\mleft(\mathvect{u}\mright), \Delta \quad \cdots \quad \Gamma \fCenter Q_h \mathvect{u}_h\mleft(\mathvect{u}\mright), \Delta$ \quad  $\Gamma \fCenter P_{j_1}\mathvect{t}_1\mleft(\mathvect{u}\mright), \Delta \quad \cdots \quad \Gamma \fCenter P_{j_{m}} \mathvect{t}_m\mleft(\mathvect{u}\mright), \Delta $}
    \rulenamelabel{$P_{i}$ R}
    \UnaryInfC{$\Gamma \fCenter P_i\mathvect{t}\mleft(\mathvect{u}\mright), \Delta$}
   \end{inlineprooftree}.
  \end{center}
  
  Next, we define the left introduction rule for the inductive predicate.
  A \emph{case distinctions} of $\Gamma , P_{i} \mathvect{u} \fCenter \Delta$ is defined as a sequent
 \[ \Gamma, \mathvect{u} = \mathvect{t} \mleft( \mathvect{y} \mright), Q_1 \mathvect{u}_1\mleft( \mathvect{y} \mright), \ldots , Q_h \mathvect{u}_h\mleft( \mathvect{y} \mright) , P_{j_1} \mathvect{t}_1  \mleft( \mathvect{y} \mright) , \ldots , P_{j_m} \mathvect{t}_m  \mleft( \mathvect{y} \mright) \vdash \Delta, \]
 where $\mathvect{y}$ is a sequence of distinct variables of the same length as $\mathvect{x}$ and 
 $y \not \in \mathFV{\Gamma \cup \Delta \cup \mathsetextension{ P_i \mathvect{u}}} $ 
 for all $y \in \mathvect{y}$, and there is a production
  \begin{center}
  \begin{inlineprooftree}
    \AxiomC{$Q_1 \mathvect{u}_1\mleft( \mathvect{x} \mright) \quad \ldots \quad Q_h \mathvect{u}_h\mleft( \mathvect{x} \mright) \quad P_{j_1} \mathvect{t}_1  \mleft( \mathvect{x} \mright)\quad \ldots \quad P_{j_m} \mathvect{t}_m  \mleft( \mathvect{x} \mright)$ }
    \UnaryInfC{$ P_{i} \mathvect{t} \mleft( \mathvect{x} \mright)$}
  \end{inlineprooftree}.
 \end{center}
 The inference rule (Case $P_{i}$) is
 \begin{center}
  \begin{inlineprooftree}
   \AxiomC{All case distinctions of $\Gamma , P_{i} \mathvect{u} \fCenter \Delta$}
   \rulenamelabel{Case $P_{i}$}
   \UnaryInfC{$\Gamma , P_{i} \mathvect{u} \fCenter \Delta $}
  \end{inlineprooftree}.
 \end{center}
 The formulas 
 $P_{j_{1}} \mathvect{t}_{1}  \mleft( \mathvect{y} \mright), \ldots, P_{j_{m}} \mathvect{t}_{m}  \mleft( \mathvect{y} \mright)$ 
 in case distinctions are said to be \emph{case-descendants} of 
 the principal formula $P_{i}\mathvect{u}$. 
 
 \begin{ex}
  \label{ex:rules_for_add}
  Let $\mathindAddonesy$, $\mathindAddtwosy$ be inductive predicates of arity three,
  $0$ be a constant symbol, and $s$ be a function symbol of arity one. 
  
  We define the productions of 
  $\mathindAddonesy$, $\mathindAddtwosy$ by  \medskip

  \hfil 
  \begin{inlineprooftree}
      \AxiomC{\phantom{$\mathindAddone 0yy$}}
      \UnaryInfC{$\mathindAddone 0yy$}
  \end{inlineprooftree},
  \qquad
  \begin{inlineprooftree}
   \AxiomC{$\mathindAddone xyz$}
   \UnaryInfC{$\mathindAddone{sx}{y}{sz}$}
  \end{inlineprooftree},
  \qquad
  \begin{inlineprooftree}
   \AxiomC{\phantom{$\mathindAddtwo 0yy$}}
   \UnaryInfC{$\mathindAddtwo 0yy$}
  \end{inlineprooftree},
  \qquad
  \begin{inlineprooftree}
   \AxiomC{$\mathindAddtwo{x}{sy}{z}$}
   \UnaryInfC{$\mathindAddtwo{sx}{y}{z}$}
  \end{inlineprooftree}. \medskip
  
  The inference rules for $\mathindAddonesy$, $\mathindAddtwosy$ are 
  
  \hfil
  \begin{tabular}{cc}
   \begin{inlineprooftree}
    \AxiomC{\phantom{$\fCenter \mathindAddone 0bb$}}
    \rulenamelabel{$\mathindAddonesy$ R${}_1$}
    \UnaryInfC{$\Gamma \fCenter \mathindAddone{0}{b}{b}, \Delta$}
   \end{inlineprooftree},
   &
       \begin{inlineprooftree}
	\AxiomC{ $\Gamma \fCenter \Delta, \mathindAddone{a}{b}{c}$}
	\rulenamelabel{$\mathindAddonesy$ R${}_2$}
	\UnaryInfC{$\Gamma \fCenter \Delta, \mathindAddone{sa}{b}{sc}$}
       \end{inlineprooftree}, \medskip
       \\ 
   \begin{inlineprooftree}
    \AxiomC{\phantom{$\fCenter \mathindAddtwo 0bb$}}
    \rulenamelabel{$\mathindAddtwosy$ R${}_1$}
    \UnaryInfC{$\Gamma \fCenter \mathindAddtwo{0}{b}{b},\Delta $}
   \end{inlineprooftree},
   & 
       \begin{inlineprooftree}
	\AxiomC{ $\Gamma \fCenter \Delta, \mathindAddtwo{a}{sb}{c}$}
	\rulenamelabel{$\mathindAddtwosy$ R${}_2$}
	\UnaryInfC{$\Gamma \fCenter \Delta, \mathindAddtwo{sa}{b}{c}$}
       \end{inlineprooftree},
  \end{tabular}\medskip
  
  \hfil 
  \begin{inlineprooftree}
  \AxiomC{$\Gamma, a=0, b=y, c=y \fCenter \Delta$}
  \AxiomC{$\Gamma, a=sx , b=y, c=sz, \mathindAddone{x}{y}{z} \fCenter \Delta$}
  \rulenamelabel{Case $\mathindAddonesy$}
  \BinaryInfC{ $\Gamma, \mathindAddone{a}{b}{c} \fCenter \Delta$}
  \end{inlineprooftree} \medskip \\ 
  ($x$, $y$, $z \not \in \mathFV{\Gamma \cup \Delta \cup \mathsetextension{\mathindAddone{a}{b}{c}}}$ 
  and $x$, $y$, $z$ are all distinct) and \medskip
  
  \hfil
  \begin{inlineprooftree}
   \AxiomC{$\Gamma, a=0, b=y, c=y \fCenter \Delta$}
   \AxiomC{$\Gamma, a=sx, b=y, c=z, \mathindAddtwo{x}{sy}{z} \fCenter \Delta$}
   \rulenamelabel{Case $\mathindAddtwosy$}
   \BinaryInfC{ $\Gamma, \mathindAddtwo{a}{b}{c} \fCenter \Delta$}
  \end{inlineprooftree} \medskip  \\ 
  ($x$, $y$, $z \not \in \mathFV{\Gamma \cup \Delta \cup \mathsetextension{\mathindAddtwo{a}{b}{c}}}$ 
  and $x$, $y$, $z$ are all distinct).
 \end{ex}
  
 %%%%%%%%%%%%%

  \subsection{Infinitary proof system \LKIDOm}
  \label{subsec:LKIDOm}
  In this section,
  we define an infinitary proof system for first-order logic with inductive definitions \LKIDOm.
  The inference rules of \LKIDOm\ are the same as that of \CLKID.

  \begin{definition}[Derivation tree]
   Let $\mathRuleSet$ be the set of names for the inference rules of \CLKID.
   Let $\mathSeqSet$ be the set of sequents.
   $\mathnatKC$ denotes the set of finite sequences of natural numbers.
   We write $\mleft\langle n_{1}, \dots, n_{k} \mright\rangle$ 
   for the sequence of the numbers $n_{1}, \dots, n_{k}$.
   We write $\sigma_{1}\sigma_{2}$ for the concatenation of $\sigma_{1}$ and $\sigma_{2}$
   with $\sigma_{1}$, $\sigma_{2}\in\mathnatKC$.
   We write $\sigma n$ for $\sigma \mleft\langle n \mright\rangle$ 
   for $\sigma\in\mathnatKC$ and $n\in\mathnat$.
   We define a \emph{derivation tree} to be a partial function
   $\mathdefparfunc{\mathdertree{D}}{\mathnatKC}{\mathSeqSet\times\mleft(\mathRuleSet\cup\mathsetextension{\text{\rulename{Bud}}}\mright)}$
   satisfying the following conditions:
   \begin{enumerate}
    \item $\mathof{\mathdom}{\mathdertree{D}}$ is prefixed-closed, that is to say,
	  if $\sigma_{1}\sigma_{2}\in\mathof{\mathdom}{\mathdertree{D}}$ for $\sigma_{1}$, $\sigma_{2}\in\mathnatKC$, 
	  then $\sigma_{1}\in\mathof{\mathdom}{\mathdertree{D}}$.
    \item If $\sigma n\in\mathof{\mathdom}{\mathdertree{D}}$ for $\sigma\in\mathnatKC$ and $n\in\mathnat$,  
	  then $\sigma m\in\mathof{\mathdom}{\mathdertree{D}}$ for any $m\leq n$.
    \item Let $\mathof{\mathdertree{D}}{\sigma}=\mleft(\Gamma_{\sigma}\fCenter\Delta_{\sigma}, R_{\sigma}\mright)$.
    \begin{enumerate}
     \item If $R_{\sigma}=\text{\rulename{Bud}}$, then $\sigma 0\notin\mathof{\mathdom}{\mathdertree{D}}$.
     \item If $R_{\sigma}\neq\text{\rulename{Bud}}$, then 
	   
	   \begin{prooftree}
	    \AxiomC{$\Gamma_{\sigma 0}\fCenter\Delta_{\sigma 0}$}
	    \AxiomC{$\cdots$}
	    \AxiomC{$\Gamma_{\sigma n}\fCenter\Delta_{\sigma n}$}
	    \TrinaryInfC{$\Gamma_{\sigma}\fCenter\Delta_{\sigma}$}
	   \end{prooftree}	       
	   is a rule $R_{\sigma}$ and $\sigma\mleft(n+1\mright)\notin\mathof{\mathdom}{\mathdertree{D}}$.
    \end{enumerate}
   \end{enumerate}
  \end{definition}

  We write $\mathSeq{\mathdertree{D}}{\sigma}$ and $\mathRule{\mathdertree{D}}{\sigma}$ 
  for $\Gamma \fCenter\Delta$ and \rulename{R}, respectively, where 
  $\mathof{\mathdertree{D}}{\sigma}=\mleft(\Gamma \fCenter\Delta, {\text{\rulename{R}}}\mright)$.

  An element in the domain of a derivation tree is called a \emph{node}.
  The empty sequence as a node is called the \emph{root}.
  The node $\sigma$ is called a \emph{bud} if $\mathRule{\mathdertree{D}}{\sigma}$ is \rulename{Bud}.
  The node which is not a bud is called an \emph{inner node}.
  A derivation tree is called \emph{infinite} if the domain of the derivation tree is infinite. 

  We sometimes identify a node $\sigma$ with the sequent $\mathSeq{\mathdertree{D}}{\sigma}$.
  
  \begin{definition}[Path] 
   We define a \emph{path} in a derivation tree $\mathdertree{D}$ to be a (possibly infinite) sequence 
   $\mleft\{ \sigma_{i} \mright\}_{0\leq i < \alpha}$ of nodes in $\mathof{\mathdom}{\mathdertree{D}}$
   such that $\sigma_{i+1}=\sigma_{i}n$ for some $n\in\mathnat$
   and $\alpha\in \mathnatpos \cup \mathsetextension{\omega}$, where 
   $\mathnatpos$ is the set of positive natural numbers and $\omega$ is the least infinite ordinal.
   A finite path $\sigma_{0}, \sigma_{1}, \dots, \sigma_{n}$ is called
   \emph{a path from $\sigma_{0}$ to $\sigma_{n}$}.
   The \emph{length of a finite path} $\mleft\{ \sigma_{i} \mright\}_{0\leq i < \alpha}$ 
   is defined as $\alpha$.
   We define \emph{the height of a node} as the length of the path from the root to the node.
  \end{definition}

  We sometimes write $\mleft(\Gamma_{i} \fCenter \Delta_{i} \mright)_{0\leq i <\alpha}$
  for the path $\mleft(\sigma_{i}\mright)_{0\leq i <\alpha}$ in a derivation tree $\mathdertree{D}$
  if $\mathof{\mathdertree{D}}{\sigma_{i}}=\mleft(\Gamma_{i} \fCenter \Delta_{i}, R_{i}\mright)$.

  \begin{definition}[Trace] 
   For a path $\mleft(\Gamma_{i}\fCenter \Delta_{i}\mright)_{0\leq i <\alpha}$
   in a derivation tree $\mathdertree{D}$,
   we define a \emph{trace following}  $\mleft(\Gamma_{i}\fCenter\Delta_{i}\mright)_{0\leq i <\alpha}$ 
   to be a sequence of formulas $\mleft(\tau_{i} \mright)_{0\leq i <\alpha}$ 
   such that the following hold:
   \begin{enumerate}
    \item $\tau_{i}$ is an atomic formula with an inductive predicate in $\Gamma_{i}$. 
    \item If $\Gamma_{i} \fCenter \Delta_{i}$ is the conclusion of (\rulename{Subst}) with $\theta$, 
	  then $\tau_{i}\equiv\tau_{i+1}\mleft[\theta\mright]$.
    \item If $\Gamma_{i}\fCenter\Delta_{i}$ is the conclusion of (\rulename{$=$ L}) 
	  with the principal formula $t=u$ and
	  $\tau_{i}\equiv F\mleft[x:=t, y:=u\mright]$, then $\tau_{i+1} \equiv F\mleft[x:=u, y:=t\mright]$.
    \item If $\Gamma_{i}\fCenter \Delta_{i}$ is  the conclusion of (\rulename{Case $P_{i}$}), 
	  then either
	  \begin{itemize}
	   \item  $\tau_{i}$ is the principal formula of the rule and $\tau_{i+1}$ 
		  is a case-descendant of $\tau_{i}$, or
	   \item  $\tau_{i+1}$ is the same as $\tau_{i}$. 
	  \end{itemize}
	  
	  In the former case, $\tau_{i}$ is said to be a \emph{progress point} of the trace.
    \item If $\Gamma_{i} \fCenter \Delta_{i}$ is the conclusion of any other rules, 
	  then $\tau_{i+1}\equiv\tau_{i}$.
   \end{enumerate}
  \end{definition}
  
  \begin{definition}[Global trace condition]
   If a trace has infinitely many progress points,
   we call the trace an \emph{infinitely progressing trace}.
   If there exists an infinitely progressing trace following a tail of the path
   $\mleft( \Gamma_{i} \fCenter \Delta_{i} \mright)_{i \geq k}$ with some $k \geq 0$
   for every infinite path $\mleft( \Gamma_{i}\fCenter\Delta_{i} \mright)_{i \geq 0}$ in a derivation tree,
   we say the derivation tree satisfies the \emph{global trace condition}.
  \end{definition}

 \begin{definition}[\LKIDOm\ pre-proof]
  We define an \LKIDOm\ \emph{pre-proof}
  to be a (possibly infinite) derivation tree $\mathdertree{D}$ without buds.
  When the root is $\Gamma\fCenter\Delta$, we call $\Gamma\fCenter\Delta$ the conclusion of the proof.
 \end{definition}
 
 \begin{definition}[\LKIDOm\ proof] %
  \label{definition:LKIDOm-proof} 
  We define an \LKIDOm\ \emph{proof} to be
  an \LKIDOm\ pre-proof that satisfies the global trace condition.
 \end{definition}

 Because of the global trace condition, the soundness of \LKIDOm\ for the standard models hold
 \cite{BrotherstonPhD, Brotherston2011}.
 In other words, if there exists an \LKIDOm\ \emph{proof} of a sequent 
 $\Gamma\fCenter\Delta$, then $\Gamma \fCenter \Delta$ is valid in any standard models.
 Moreover, cut-free completeness of \LKIDOm\ for the standard models hold.
 In other words, if $\Gamma \fCenter \Delta$ is valid in any standard models,
 there exists an \LKIDOm\ cut-free \emph{proof} of $\Gamma\fCenter\Delta$ 
 \cite{BrotherstonPhD, Brotherston2011}.

  \subsection{Cyclic proof system \CLKID} \label{subsec:CLKID}
  In this section, we introduce a cyclic proof system \CLKID. 

  \begin{definition}[Companion]
   For a finite derivation tree $\mathdertree{D}$,
   we define the \emph{companion} for a bud $b$ as an inner node $\sigma$ in $\mathdertree{D}$ with 
   $\mathSeq{\mathdertree{D}}{\sigma}=\mathSeq{\mathdertree{D}}{b}$.
  \end{definition}

  \begin{definition}[\CLKID\ pre-proof]
   We define a \CLKID\ \emph{pre-proof}
   to be a pair $\mleft(\mathdertree{D},  \mathcompanion{C}\mright)$
   such that $\mathdertree{D}$ is a finite derivation tree and 
  $\mathcompanion{C}$ is a function mapping each bud to its companion.
   When the root is $\Gamma \fCenter \Delta$, 
   we call $\Gamma\fCenter\Delta$ the conclusion of the proof.
  \end{definition}
  
  \begin{definition}[Tree-unfolding]
   \label{definition:tree-unfolding}
   For a \CLKID\ pre-proof $\mleft(\mathdertree{D},  \mathcompanion{C}\mright)$,
   a \emph{tree-unfolding} $\mathtreeunfolding{\mathdertree{D},  \mathcompanion{C}}$
   of $\mleft(\mathdertree{D},  \mathcompanion{C}\mright)$ is recursively defined by
   \[ 
   \mathof{\mathtreeunfolding{\mathdertree{D},  \mathcompanion{C}}}{\sigma}=
   \begin{cases}
    \mathof{\mathdertree{D}}{\sigma}, 
    & \text{if $\sigma\in\mathof{\mathdom}{\mathdertree{D}}\setminus\mathBudset{\mathdertree{D}}$,}
    \\ %
    \mathof{\mathtreeunfolding{\mathdertree{D},  \mathcompanion{C}}}{\sigma_{3}\sigma_{2}}, 
    & \text{if $\sigma\notin\mathof{\mathdom}{D}\setminus\mathBudset{\mathdertree{D}}$ with $\sigma=\sigma_{1}\sigma_{2}$, $\sigma_{1}\in \mathBudset{\mathdertree{D}}$ and $\sigma_{3}=\mathof{\mathcompanion{C}}{\sigma_{1}}$,}
   \end{cases}
   \]
   where $\mathBudset{\mathdertree{D}}$ is the set of buds in $\mathdertree{D}$.
  \end{definition}

  Note that a tree-unfolding is an \LKIDOm\ pre-proof. 

  \begin{definition}[\CLKID\ proof]
   We define a \emph{\CLKID\ proof} of a sequent $\Gamma \fCenter \Delta$ 
   to be a \CLKID\ pre-proof of $\Gamma \fCenter \Delta$ 
   whose tree-unfolding satisfies the global trace condition.
   If a \CLKID\ proof of $\Gamma \fCenter \Delta$ exists, 
   we say $\Gamma \fCenter \Delta$ is \emph{provable} in \CLKID.
   A \CLKID\ proof in which (\rulename{Cut}) does not occur
   is called \emph{cut-free}.
   If a cut-free \CLKID\ proof of $\Gamma \fCenter \Delta$ exists,  
   we say $\Gamma \fCenter \Delta$ is \emph{cut-free provable} in \CLKID.
  \end{definition}

  \begin{ex}%
   \label{ex:proofAddone}
   The derivation tree given in Figure \ref{fig:counter-ex-D1}
   is the proof of $\mathindAddone{x_1}{sy_1}{z_1} \fCenter \mathindAddone{sx_1}{y_1}{z_1}$ in \CLKID,  
   where ($\star$) indicates the pairing of a companion with a bud and
   the underlined formulas are the infinitely progressing trace 
   for the infinite path (some applying rules and some labels are omitted for limited space).
   \begin{figure}[thbp]
 \centering
 \begin{prooftree}
  \small
  \AxiomC{}
  \rulenamelabel{$\mathindAddonesy$ R${}_1$}
  \UnaryInfC{$\fCenter \mathindAddone{0}{y_1}{y_1}$}
  \rulenamelabel{$\mathindAddonesy$ R${}_2$}
  \UnaryInfC{$\fCenter \mathindAddone{s0}{y_1}{sy_1}$}
  % \rulenamelabel{Weak}
  \UnaryInfC{$\fCenter \mathindAddone{s0}{y_1}{sy_1}$}
  % \rulenamelabel{$=$ L}
  \UnaryInfC{$sy_1=y_2 \fCenter \mathindAddone{s0}{y_1}{y_2}$}\doubleLine
  % \rulenamelabel{$=$ L}
  \UnaryInfC{$\begin{aligned}x_1&=0, \\ sy_1&=y_2, \\ z_1&=y_2\end{aligned}\fCenter \mathindAddone{sx_1}{y_1}{z_1}$} %(3)

  \AxiomC{($\star$) $\underline{\mathindAddone{x_1}{sy_1}{z_1}} \fCenter \mathindAddone{sx_1}{y_1}{z_1}$}
  % \rulenamelabel{Subst}
  \UnaryInfC{$\underline{\mathindAddone{x_2}{sy_1}{z_2}} \fCenter \mathindAddone{sx_2}{y_1}{z_2}$}
  \rulenamelabel{$\mathindAddonesy$ R${}_2$}
  \UnaryInfC{$\underline{\mathindAddone{x_2}{sy_1}{z_2}} \fCenter \mathindAddone{ssx_2}{y_1}{sz_2}$}
  % \rulenamelabel{Weak}
  \UnaryInfC{$\underline{\mathindAddone{x_2}{sy_1}{z_2}} \fCenter \mathindAddone{ssx_2}{y_1}{sz_2}$}\doubleLine
  % \rulenamelabel{$=$ L}
  \UnaryInfC{$\begin{aligned}x_1&=sx_2,\\ sy_1&=y_2, \\ z_1&=sz_2,\end{aligned} \underline{\mathindAddone{x_2}{y_2}{z_2}} \fCenter \mathindAddone{sx_1}{y_1}{z_1}$} %(3)
  \rulenamelabel{Case $\mathindAddonesy$} %(3)
  \BinaryInfC{($\star$) $\underline{\mathindAddone{x_1}{sy_1}{z_1}} \fCenter \mathindAddone{sx_1}{y_1}{z_1}$} 
 \end{prooftree}
  \caption{A \CLKID\ proof} 
 \label{fig:counter-ex-D1}
\end{figure}
  \end{ex}

\subsection{Cycle-normalization}\label{subsec:cycle-normalization}

\newcommand{\Dom}{\mathrm{dom}}
\def\Bar{\overline}
\def\prove{\mathrel \vdash}

\let\TmpA=\[
\let\TmpB=\]
\def\[{\TmpA\begin{array}{l}}
\def\]{\end{array}\TmpB}

This section proves \emph{cycle-normalization} for \CLKID. 
It is proved in \cite{BrotherstonPhD}, but we will give a much shorter proof.

A \CLKID\ pre-proof in which each companion is an ancestor of the corresponding bud is called
\emph{cycle-normal}.
The following proposition states the cycle-normalization holds for \CLKID. 

\begin{prop}\label{prop:cyclenormalization}
 For a \CLKID\ pre-proof $(D,C)$, we have
 a \CLKID\ cycle-normal pre-proof $(D',C')$ such that
 the tree-unfolding of $(D,C)$ is that of $(D',C')$.
\end{prop}

% {\em Proof.}
\begin{proof}
 We write $\sigma \subseteq \sigma'$ when $\sigma$ is an initial segment
 of $\sigma'$. We write $|\sigma|$ for the length of a sequence $\sigma$.
 We define $D^{(\sigma)}$ by $D^{(\sigma)}(\sigma_1) = D(\sigma \sigma_1)$,
 $\Bar S$ as $\{ \sigma' \ |\ \sigma' \subseteq \sigma \in S \}$,
 and
 $S^\circ$ as $\{ \sigma' \ |\ \sigma' \subsetneqq \sigma \in S \}$.

 Let $D_1$ be the tree unfolding of $(D,C)$.

 Define
 \[
 S_1 = \{ \sigma \in \mathof{\Dom}{D_1} \ |\ \exists \sigma' \subsetneqq \sigma
 (D_1^{(\sigma)} = D_1^{(\sigma')}),
 \forall \sigma_1\subsetneqq \sigma \forall \sigma_2 \subsetneqq \sigma(D_1^{(\sigma_1)} \ne D_1^{(\sigma_2)}),
 \forall n\exists \sigma_1 \supseteq \sigma(\sigma_1 \in \Dom(D_1), |\sigma_1|\ge n) \}, \\

 S_2 = \{ \sigma \in \Dom(D_1) \ |\ \sigma 0 \notin \Dom(D_1), 
 \forall \sigma' \subseteq \sigma(\sigma' \notin S_1) \}.
 \]

 $S_1$ is the set of nodes such that
 the node is on some infinite path and
 the node is of the smallest height on the path
 among nodes, each of which has some inner node 
 of the same subtree.
 $S_2$ is the set of leaf nodes of finite paths
 which are not cut by $S_1$.

 Define $D'$ by
 \[
 D'(\sigma) = D_1(\sigma) 
 \hbox{ if $\sigma \in (S_1)^\circ \cup \Bar{S_2}$}, \\

 D'(\sigma) = (\Gamma \prove \Delta,\mathrm{Bud}) \hbox{ 
 if $\sigma \in S_1, D_1(\sigma) = (\Gamma \prove \Delta, R)$}.
 \]

 Define $C'$ by $C'(\sigma) = \sigma'$ for a bud $\sigma$ of $D'$
 where $\sigma' \subsetneqq \sigma, D_1^{(\sigma)} = D_1^{(\sigma')}$.

 We can show that $\Dom(D')$ is finite as follows.
 Since $\Dom(D') = \Bar{S_1} \cup \Bar{S_2}$,
 we have $\Dom(D') \subseteq \Dom(D_1)$.
 Since $D_1$ is finite-branching, $D'$ is so.
 Assume $\Dom(D')$ is infinite to show contradiction.
 By K\"onig's lemma, there is some infinite path $(\sigma_i)_i$
 such that $\sigma_i \in \Dom(D')$.
 Since $D_1$ is regular, the set $\{ D_1^{(\sigma_i)} \}_i$ is finite.
 Hence there are $j<k$ such that $D_1^{(\sigma_j)} = D_1^{(\sigma_k)}$.
 Take the smallest $k$ among such $k$'s.
 Then $\sigma_k \in S_1$.
 Hence $\sigma_{k+1} \notin \Bar{S_1}$.
 Hence $\sigma_{k+1} \notin \Dom(D')$, which contradicts.

 Then $(D',C')$ is a \CLKID\ cycle-normal pre-proof.

 Define $D_1'$ as the tree-unfolding of $(D',C')$.

 We can show $D_1=D_1'$ on $\Dom(D_1')$ as follows.

 Case 1 where for any $\sigma' \subseteq \sigma$, $\sigma' \notin S_1$.
 $D_1'(\sigma) = D'(\sigma) = D_1(\sigma)$.

 Case 2 where 
 there is some $\sigma_1 \subseteq \sigma$ such that $\sigma_1 \in S_1$.
 Let $\sigma_1 \sigma_2$ be $\sigma$ and $\sigma_3$ be $C'(\sigma_1)$.
 Then $D_1(\sigma) = D_1^{(\sigma_1)}(\sigma_2) = D_1^{(\sigma_3)}(\sigma_2)
 = D_1(\sigma_3 \sigma_2) = D_1'(\sigma_3 \sigma_2)$ by the induction hypothesis,
 it is $D_1'(\sigma_1 \sigma_2)$ by definition of $D_1'$, 
 and it is $D_1'(\sigma)$.

 We can show $\Dom(D_1) \subseteq \Dom(D_1')$ as follows.
 By induction on $|\sigma|$, we will show
 $\sigma \in \Dom(D_1) \hbox{ implies } \sigma \in \Dom(D_1')$.
 If $\sigma \in S_1^\circ \cup \Bar{S_2}$, then
 $\sigma \in \Dom(D') - S_1$.
 Hence $\sigma \in \Dom(D_1')$.
 If there is some $\sigma_1 \subsetneqq \sigma$ such that $\sigma_1 \in S_1$,
 then by letting $\sigma = \sigma_1 \sigma_2$ and $\sigma_3 = C'(\sigma_1)$,
 $D_1(\sigma) = D_1(\sigma_3 \sigma_2)$ by definition of $C'$,
 by the induction hypothesis for $\sigma_3 \sigma_2$ it is $D_1'(\sigma_3 \sigma_2)$,
 and it is $D_1'(\sigma)$ by definition of $D_1'$.
 Hence we have shown $\Dom(D_1) \subseteq \Dom(D_1')$.

Hence $D_1=D_1'$.
\end{proof}

% $\Box$
% }
%#!latexmk -c -gg -lualatex main.tex
 \section{A counterexample to cut-elimination in \CLKID}
 \label{sec:main}
 In this section,  we prove the following theorem, which is the main theorem.

 \begin{theorem} 
  \label{thm:main}
  Let $0$ be a constant symbol, 
  $s$ be a function symbol of arity one, and 
  $\mathindAddonesy$ and $\mathindAddtwosy$ be inductive predicates of arity three
  with the following productions: \medskip
  
  \begin{inlineprooftree}
      \AxiomC{\phantom{$\mathindAddone 0yy$}}
      \UnaryInfC{$\mathindAddone 0yy$}
  \end{inlineprooftree},
  \qquad
  \begin{inlineprooftree}
   \AxiomC{$\mathindAddone xyz$}
   \UnaryInfC{$\mathindAddone{sx}{y}{sz}$}
  \end{inlineprooftree},
  \qquad
  \begin{inlineprooftree}
   \AxiomC{\phantom{$\mathindAddtwo 0yy$}}
   \UnaryInfC{$\mathindAddtwo 0yy$}
  \end{inlineprooftree},
  \qquad
  \begin{inlineprooftree}
   \AxiomC{$\mathindAddtwo{x}{sy}{z}$}
   \UnaryInfC{$\mathindAddtwo{sx}{y}{z}$}
  \end{inlineprooftree}. \medskip

  \begin{enumerate}
   \item $\mathindAddtwo{x}{y}{z} \fCenter \mathindAddone{x}{y}{z}$ is provable in \CLKID. 
	 \label{item:thm-main-provable}
   \item $\mathindAddtwo{x}{y}{z} \fCenter \mathindAddone{x}{y}{z}$ is not cut-free provable in \CLKID.
	 \label{item:thm-main-not-cut-free-provable}
  \end{enumerate}
 \end{theorem}

 This theorem means that
 $\mathindAddtwo{x}{y}{z} \fCenter \mathindAddone{x}{y}{z}$ is a counterexample to cut-elimination in \CLKID.

 Note that $\mathindAddonesy$ and $\mathindAddtwosy$ in the theorem are the same predicates
 in Example \ref{ex:rules_for_add}.

 \subsection{The outline of the proof}
 Before proving the theorem, 
 we outline our proof for \ref{item:thm-main-not-cut-free-provable} of the theorem.

 Assume there exists a cut-free \CLKID\ proof of
 $\mathindAddtwo{x}{y}{z} \fCenter \mathindAddone{x}{y}{z}$.
 Because of a technical issue, we use \CLKIDa, 
 a cyclic proof system with the same provability as \CLKID\ 
 whose inference rules are the same as \CLKID\
 except for the rule (\rulename{$=$ L}) (Definition \ref{def:CLKIDa}).
 We show that there exists a cut-free cycle-normal \CLKIDa\ proof of the sequent 
 (Proposition \ref{prop:CLKIDa}).
 Let $\mathproofcf$ be the \CLKIDa\ proof. 

 Next, we define the relation $\mathequivrel{\Gamma}$ for a sequent $\Gamma\fCenter\Delta$ 
 to be the smallest congruence relation on terms containing $t_{1} = t_{2} \in \Gamma$ 
 (Definition \ref{def:equivrel}).
 Then, we define the index of $\mathindAddtwo{a}{b}{c}$ in a sequent $\Gamma\fCenter\Delta$
 (Definition \ref{def:index}). 
 %For simplicity, we write $s^{n}x$ for $\overbrace{s \cdots s}^{n}x$. %??? 直後に使っている
 If there uniquely exists $n-m$ 
 such that $n$,  $m \in \mathnat$  and $s^{n} b \mathequivrel{\Gamma} s^{m} b'$ 
 for some $\mathindAddone{a'}{b'}{c'} \in \Delta$,
 then the index of $\mathindAddtwo{a}{b}{c}$ is defined as $m-n$.
 If $s^{n} b \not\mathequivrel{\Gamma} s^{m} b'$
 for any $\mathindAddone{a'}{b'}{c'} \in \Delta$ and any $n$,  $m \in \mathnat$,
 the index of $\mathindAddtwo{a}{b}{c}$ is defined as $\bot$.
 The index of $\mathindAddtwo{a}{b}{c}$ may be undefined,
 but the index is always defined in a special sequent, called an \emph{index sequent} 
 (Definition \ref{def:annoying}).
 A \emph{switching point} is defined as a node 
 that is the conclusion of (\rulename{Case $\mathindAddtwosy$}) with
 the principal formula whose index is $\bot$ (Definition \ref{def:bad_app}).
 An \emph{index path} is defined as
 a path $\mleft( \Gamma_{i}\fCenter\Delta_{i} \mright)_{0\leq i <\alpha}$ 
 of $\mathtreeunfolding{\mathproofcf}$ such that
 $\Gamma_{0}\fCenter\Delta_{0}$ is an index sequent and 
 $\Gamma_{i}\fCenter\Delta_{i}$ is a switching point
 if $\Gamma_{i+1}\fCenter\Delta_{i+1}$ is the left assumption of $\Gamma_{i}\fCenter\Delta_{i}$ 
 (Definition \ref{def:bad-path}).
 Then, we have
 \outlinelemma The root is an index sequent. \label{outlinelemma:root}
 \outlinelemma Every sequent in an index path is an index sequent (Lemma \ref{lemma:annoying}). \label{outlinelemma:index-path}
 \outlinelemma There exists a switching point on an infinite index path (Lemma \ref{lemma:key_lemma}). \label{outlinelemma:key}
 \outlinelemma The rightmost path from an index sequent is infinite (Lemma \ref{lemma:rightmost}). \label{outlinelemma:rightmost-path}
  
  At last, 
  we show there exist infinite nodes in the derivation tree $\mathdertreecf$.
  Because of (\ref{outlinelemma:root}) and (\ref{outlinelemma:rightmost-path}), the rightmost path from the root is an infinite index path. 
  By (\ref{outlinelemma:key}), there exists a switching point on the path.
  Let $c_{0}$ be the node of the smallest height among such switching points.
  Let $a_{0}$ be the left assumption of $c_{0}$.
  By (\ref{outlinelemma:index-path}), the sequent of $a_{0}$ is an index sequent.
  By (\ref{outlinelemma:rightmost-path}), the rightmost path from $a_{0}$ is infinite.
  Therefore, there exists a bud $b_{0}$ in the rightmost path from $a_{0}$.
  By (\ref{outlinelemma:key}) and the definition of $c_{0}$, 
  there exists a switching point between $a_{0}$ and $b_{0}$ .
  Let $c_{1}$ be the node of the smallest height among such switching points.
  The nodes $c_{0}$ and $c_{1}$ are distinct by their definitions.
  We repeat this process as in Figure \ref{fig:idea}. 
  Finally, we get a set of infinite nodes $\mathsetintension{c_{i}}{i\in\mathnat}$.
  This is a contradiction since the set of nodes of $\mathdertreecf$ is finite.
  
  \begin{figure}[tbhp]
 \begin{prooftree}
	   \AxiomC{$\begin{aligned}&\vdots\\&a_{2}\tikzmark{node-a2}\end{aligned}$}%(4)
	   \AxiomC{$\begin{aligned}&b_{1}\tikzmark{node-b1}\\&\vdots\end{aligned}$}
  \rulenamelabel{Case $\mathindAddtwosy$\tikzmark{node-M2}}
  \BinaryInfC{$c_{2}$\tikzmark{node-c2}}%(4)
  \noLine
	      \UnaryInfC{$\begin{aligned}&\vdots\\&a_{1}\tikzmark{node-a1}\end{aligned}$}%(3)
  
	   \AxiomC{$\begin{aligned}&b_{0}\tikzmark{node-b0}\\&\vdots\end{aligned}$} %(3)
  
  %(3)
  \rulenamelabel{Case $\mathindAddtwosy$\tikzmark{node-M1}}
  \BinaryInfC{$c_{1}$\tikzmark{node-c1}}
  \noLine
	      \UnaryInfC{$\begin{aligned}&\vdots\\&a_{0}\tikzmark{node-a0}\end{aligned}$}%(2)
	   \AxiomC{$\begin{aligned}&b\tikzmark{node-b}\\&\vdots\end{aligned}$}  %(2)

  \rulenamelabel{Case $\mathindAddtwosy$\tikzmark{node-Mr}}
	       \BinaryInfC{$\begin{aligned}&c_{0}\tikzmark{node-c0}\\&\vdots\end{aligned}$}
  \noLine
  \UnaryInfC{$\mathindAddtwo{x}{y}{z} \fCenter$ \tikzmark{node-root} $\mathindAddone{x}{y}{z}$}
 \end{prooftree}
 \begin{tikzpicture}[remember picture, overlay, thick, relative, auto, line width=0.3pt]
  \coordinate (b) at ({pic cs:node-b});
  \coordinate (root) at ({pic cs:node-root});
  \coordinate (M) at ({pic cs:node-Mr});
  \coordinate (c0) at ({pic cs:node-c0});
  \coordinate (b0) at ({pic cs:node-b0});
  \coordinate (a0) at ({pic cs:node-a0});
  \coordinate (M1) at ({pic cs:node-M1});
  \coordinate (c1) at ({pic cs:node-c1});
  \coordinate (b1) at ({pic cs:node-b1});
  \coordinate (a1) at ({pic cs:node-a1});
  \coordinate (M2) at ({pic cs:node-M2});
  \coordinate (c2) at ({pic cs:node-c2});
  
  \draw[->]($(b)$)..controls ($(M)+(5.0em, 0)$) .. ($(c0)!0.5!(root)$);

  \draw[->] ($(b0)$)..controls ($(M1)+(5.0em, 0)$) .. ($(c1)!0.5!(a0)$);

  \draw[->] ($(b1)$) ..controls ($(M2)+(5.0em, 0)$) .. ($(c2)!0.5!(a1)$);
 \end{tikzpicture}
 \caption{Construction of $\mleft(c_{i}\mright)_{i\in\mathnat}$} 
 \label{fig:idea}
\end{figure}

  \subsection{Another cyclic proof system \CLKIDa}
  We give some definitions and lemmas for proving \ref{item:thm-main-not-cut-free-provable}
  of Theorem \ref{thm:main}.
  We consider a cyclic proof system \CLKIDa, which is obtained by changing 
  the left introduction rule for ``='' slightly.
  \begin{definition}[\CLKIDa]
   \label{def:CLKIDa}
   \CLKIDa\ is the cyclic proof system obtained by replacing (\rulename{$=$ L}) with
   \begin{center}
    \begin{inlineprooftree}
     \AxiomC{$\Gamma \mleft[x := u, y := t\mright], t = u \fCenter \Delta \mleft[x := u, y := t\mright]$}
     \RightLabel{(\rulenameLa)}
     \UnaryInfC{$\Gamma \mleft[ x := t, y := u \mright], t = u \fCenter \Delta \mleft[x := t, y := u\mright]$}
    \end{inlineprooftree}.
   \end{center}
   The principal formula of the rule  of (\rulenameLa) is defined as $t=u$.

   Definitions of derivation trees, companions, pre-proofs, proofs for \CLKIDa\ are similar to \CLKID. 
  \end{definition} 
 
  The provability of \CLKID\ is the same as that of \CLKIDa, 
  since (\rulename{$=$ L}) is derivable in \CLKIDa\ by 
  \begin{center}
   \begin{inlineprooftree}
    \AxiomC{$\Gamma \mleft[x := u, y := t\mright] \fCenter \Delta \mleft[x := u, y := t\mright]$}
    \rulenamelabel{Weak}
    \UnaryInfC{$\Gamma \mleft[x := u, y := t\mright], t = u \fCenter \Delta \mleft[x := u, y := t\mright]$}
    \RightLabel{(\rulenameLa)}
    \UnaryInfC{$\Gamma \mleft[ x := t, y := u \mright], t = u \fCenter \Delta \mleft[x := t, y := u\mright]$}
   \end{inlineprooftree}.
  \end{center}

  \CLKIDa\ is necessary because of Lemma \ref{lemma:index} \ref{item:lemma-index_not_change}.
  For \CLKID, Lemma \ref{lemma:index} \ref{item:lemma-index_not_change} does not hold.
  \begin{prop}%
   \label{prop:CLKIDa}
   If there exists a cut-free \CLKID\ proof of $\Gamma \fCenter \Delta$, 
   then there exists a cut-free cycle-normal \CLKIDa\ proof of $\Gamma \fCenter \Delta$.
  \end{prop}

  \begin{proof} 
   Let $\mathprooffig{P}_{0}$ be a cut-free \CLKID\ proof of $\Gamma \fCenter \Delta$.
   By Proposition \ref{prop:cyclenormalization},
   there exists a cycle-normal \CLKID\ pre-proof  $\mathprooffig{P}_{1}$
   whose tree-unfolding is the same as that of $\mathprooffig{P}_{0}$. 
   Since the tree-unfolding of $\mathprooffig{P}_{1}$ satisfies the global trace condition,
   $\mathprooffig{P}_{1}$ is a cycle-normal cut-free \CLKID\ proof of $\Gamma \fCenter \Delta$.
  
   A cut-free cycle-normal \CLKID\ proof of $\Gamma \fCenter \Delta$ is
   transformed into the \CLKIDa\ proof of $\Gamma \fCenter \Delta$ by replacing all applications
   from (\rulename{$=$ L}) to (\rulenameLa) and weakening.
   Since this replacement does not change the rules except (\rulename{$=$ L}) in the \CLKID\ proof
   and the sequents of buds and companions, the obtained \CLKIDa\ proof is cut-free and cycle-normal.
  \end{proof}

  \subsection{Assuming cut-free proof} \label{subsec:assuming}
  In Sections \ref{subsec:assuming}, \ref{subsec:Equality}, \ref{subsec:Index} and \ref{subsec:proof},
  we assume there exists a cut-free \CLKID\ proof of 
  $\mathindAddtwo{x}{y}{z} \fCenter \mathindAddone{x}{y}{z}$
  for contradiction.
  By Proposition \ref{prop:CLKIDa}, there exists a cut-free cycle-normal \CLKIDa\ proof of 
  $\mathindAddtwo{x}{y}{z} \fCenter \mathindAddone{x}{y}{z}$.
  \emph{We write $\mleft(\mathdertreecf, \mathcompanioncf\mright)$ 
  for a cut-free cycle-normal \CLKIDa\ proof of $\mathindAddtwo{x}{y}{z} \fCenter \mathindAddone{x}{y}{z}$.}

 \begin{rem}%
  \label{rem:cut-freeness}
  Let $\Gamma \fCenter \Delta$ be a sequent in $\mathproofcf$.
  By induction on the height of sequents in $\mathdertreecf$, we can easily show the following statements:
  
   \begin{enumerate}%
    \item $\Gamma$ consists of only atomic formulas with $=$, $\mathindAddtwosy$.
    \item $\Delta$ consists of only atomic formulas with $\mathindAddonesy$.
    \item A term in $\Gamma$ and $\Delta$ is of the form $s^{n}0$ or $s^{n}x$ with some variable $x$.
	  \label{s-term}
    \item The possible rules in $\mathproofcf$ are
	  (\rulename{Weak}), (\rulename{Subst}), (\rulenameLa), 
	  (\rulename{Case $\mathindAddtwosy$}),	(\rulename{$\mathindAddonesy$ R${}_1$}) and
	  (\rulename{$\mathindAddonesy$ R${}_2$}).
   \end{enumerate}
 \end{rem}

 By \ref{s-term}, without loss of generality, we can assume terms
 in Sections \ref{subsec:Equality}, \ref{subsec:Index} and \ref{subsec:proof} 
 are of the form $s^{n}0$ or $s^{n}x$ with some variable $x$.
  
  \subsection{Equality in a sequent} \label{subsec:Equality}
  In this section, we define the equality $\mathequivrel{\Gamma}$ in a sequent $\Gamma\fCenter\Delta$ 
  and show some properties.

 \begin{definition}[$\mathequivrel{\Gamma}$]%
  \label{def:equivrel}
  For a set of formulas $\Gamma$,
  we define the relation $\mathequivrel{\Gamma}$ to be the smallest congruence relation on terms
  which satisfies the condition that $t_1 = t_2 \in \Gamma$ implies $t_1 \mathequivrel{\Gamma} t_2$.
 \end{definition}
 
 \begin{definition}[$\mathdeprel\Gamma$]%
  \label{def:deprel}
  For a set of formulas $\Gamma$ and terms $t_{1}$, $t_{2}$,
  we define $t_{1} \mathdeprel{\Gamma} t_{2}$
  by $s^{n} t_{1} \mathequivrel{\Gamma} s^{m} t_{2}$ for some  $n, m \in \mathnat$.
 \end{definition}

 For a term $t$, we define $\mathvars{t}$ as a variable or a constant in $t$.
 Note that $\mathdeprel{\Gamma}$ is a congruence relation and
 also note that $t \mathdeprel{\Gamma} u$ if $\mathvars{t}=\mathvars{u}$.

 \begin{lem}%
  \label{lemma:subst_rel}
  Let $\Gamma$ be a set of formulas and $\theta$ be a substitution.
  \begin{enumerate}
   \item For any terms $t_{1}$ and $t_{2}$,
	 $t_{1}\mleft[\theta\mright] \mathequivrel{\Gamma\mleft[\theta\mright]} t_{2}\mleft[\theta\mright]$ 
	 if $t_{1} \mathequivrel{\Gamma} t_{2}$. 
	 \label{item:lemma-subst_rel-equiv}
   \item For any terms $t_{1}$ and $t_{2}$,
	 $t_{1} \not\mathdeprel{\Gamma} t_{2}$
	 if $t_{1}\mleft[\theta\mright] \not\mathdeprel{\Gamma\mleft[\theta\mright]} t_{2}\mleft[\theta\mright]$. 
	 \label{item:lemma-subst_rel-deprel}
  \end{enumerate}
 \end{lem}
 
 \begin{proof}
  \ref{item:lemma-subst_rel-equiv}
  We prove the statement by induction on the definition of $\mathequivrel{\Gamma}$. %??? Gamma の構成についての帰納法を回しています
  We only show the base case.
  Assume ${t_{1}=t_{2}}\in\Gamma$.
  Then, ${t_{1}\mleft[\theta\mright]=t_{2}\mleft[\theta\mright]}\in\Gamma\mleft[\theta\mright]$.
  Thus, $t_{1}\mleft[\theta\mright] \mathequivrel{\Gamma\mleft[\theta\mright]} t_{2}\mleft[\theta\mright]$.

  \ref{item:lemma-subst_rel-deprel}
  By Definition \ref{def:deprel} and \ref{item:lemma-subst_rel-equiv}, we have the statement.
 \end{proof}

 \begin{lem}%
  \label{lemma:eq_rel}
  Let $\Gamma$ be a set of formulas, $u_{1}$, $u_{2}$ be terms, $v_{1}$, $v_{2}$ be variables, 
  $\Gamma_{1} \equiv \mleft(\Gamma\mleft[v_{1}:=u_{1}, v_{2}:=u_{2}\mright], u_{1}=u_{2}\mright)$, and 
  $\Gamma_{2} \equiv \mleft(\Gamma\mleft[v_{1}:=u_{2}, v_{2}:=u_{1}\mright], u_{1}=u_{2}\mright)$.
  \begin{enumerate}
   \item For any terms $t_{1}$ and $t_{2}$,
	 $t_{1}\mleft[v_{1}:=u_{1}, v_{2}:=u_{2}\mright] \mathequivrel{\Gamma_{1}} t_{2}\mleft[v_{1}:=u_{1}, v_{2}:=u_{2}\mright]$ 
	 if  
	 $t_{1}\mleft[v_{1}:=u_{2}, v_{2}:=u_{1}\mright] \mathequivrel{\Gamma_{2}} t_{2}\mleft[v_{1}:=u_{2}, v_{2}:=u_{1}\mright]$. 
	 \label{item:lemma-eq_rel-equiv}
   \item For any terms $t_{1}$ and $t_{2}$,
	 $t_{1}\mleft[v_{1}:=u_{2}, v_{2}:=u_{1}\mright] \not\mathdeprel{\Gamma_{2}} t_{2}\mleft[v_{1}:=u_{2}, v_{2}:=u_{1}\mright]$ 
	 if 
	 $t_{1}\mleft[v_{1}:=u_{1}, v_{2}:=u_{2}\mright] \not\mathdeprel{\Gamma_{1}} t_{2}\mleft[v_{1}:=u_{1}, v_{2}:=u_{2}\mright]$. \label{item:lemma-eq_rel-deprel}
  \end{enumerate}
 \end{lem}

 \begin{proof}%
  \ref{item:lemma-eq_rel-equiv}
  We prove the statement by induction on the definition of $\mathequivrel{\Gamma_{2}}$.
  We only show the base case.
  Assume
  ${t_{1}\mleft[v_{1}:=u_{2}, v_{2}:=u_{1}\mright]=t_{2}\mleft[v_{1}:=u_{2}, v_{2}:=u_{1}\mright]} \in \Gamma_{2}$
  to show $t_{1}\mleft[v_{1}:=u_{1}, v_{2}:=u_{2}\mright] \mathequivrel{\Gamma_{1}} t_{2}\mleft[v_{1}:=u_{1}, v_{2}:=u_{2}\mright]$.
  If $t_{1}\mleft[v_{1}:=u_{2}, v_{2}:=u_{1}\mright]=t_{2}\mleft[v_{1}:=u_{2}, v_{2}:=u_{1}\mright]$
  is $u_{1}=u_{2}$, then
  $t_{1}=t_{2}$ is $v_{2}=v_{1}$, $v_{2}=u_{2}$, $u_{1}=v_{1}$, or $u_{1}=u_{2}$.
  Therefore,
  $t_{1}\mleft[v_{1}:=u_{1}, v_{2}:=u_{2}\mright] \mathequivrel{\Gamma_{1}} t_{2}\mleft[v_{1}:=u_{1}, v_{2}:=u_{2}\mright]$.

  Assume $t_{1}\mleft[v_{1}:=u_{2}, v_{2}:=u_{1}\mright]=t_{2}\mleft[v_{1}:=u_{2}, v_{2}:=u_{1}\mright]$
  is not $u_{1}=u_{2}$.
  By case analysis, we have $t_{1}=t_{2}\in\Gamma$. Hence,
  ${t_{1}\mleft[v_{1}:=u_{1}, v_{2}:=u_{2}\mright]=t_{2}\mleft[v_{1}:=u_{1}, v_{2}:=u_{2}\mright]}\in \Gamma_{1}$. 
  Therefore, we have
  $t_{1}\mleft[v_{1}:=u_{1}, v_{2}:=u_{2}\mright]\mathequivrel{\Gamma_{1}}t_{2}\mleft[v_{1}:=u_{1}, v_{2}:=u_{2}\mright]$.

  \ref{item:lemma-eq_rel-deprel}
  By Definition \ref{def:deprel} and \ref{item:lemma-eq_rel-equiv}, we have the statement.
 \end{proof}

 \begin{lem}%
  \label{lemma:strcl}
  For a set of formulas $\Gamma$,
  the following statements are equivalent:
  \begin{enumerate}
   \item $u_{1}\mathequivrel{\Gamma}u_{2}$. \label{item:lemma-strcl-equivrel}
   \item There exists a finite sequence of terms 
	 $\mleft( t_{i} \mright)_{0\leq i \leq n}$ with $n\geq 0$ such that 
	 $t_{0}\equiv u_{1}$, $t_{n}\equiv u_{2}$ and ${t_{i}=t_{i+1}}\in\mleft[\Gamma\mright]$ 
	 for $0\leq i < n$, where
	 \[
	 \mleft[\Gamma\mright]=\mathsetintension{s^{n}t_{1}=s^{n}t_{2}}{
	 n\in\mathnat \text{ and either }
	 t_{1}=t_{2}\in \Gamma \text{ or } t_{2}=t_{1}\in\Gamma}.
	 \]
	 \label{item:lemma-strcl}
  \end{enumerate}

 \end{lem}

 \begin{proof}%
  \ref{item:lemma-strcl-equivrel} $\Rightarrow$ \ref{item:lemma-strcl}:
  Assume $u_{1}\mathequivrel{\Gamma}u_{2}$
  to prove \ref{item:lemma-strcl} by induction on the definition of  $\mathequivrel{\Gamma}$.
  We consider cases according to the clauses of the definition. 

  Case 1.
  If $u_{1}=u_{2}\in \Gamma$, 
  then we have $u_{1}=u_{2}\in \mleft[\Gamma\mright]$.
  Thus, we have \ref{item:lemma-strcl}.

  Case 2.
  If $u_{1}\equiv u_{2}$,
  then we have \ref{item:lemma-strcl}.

  Case 3.
  We consider the case where
  $u_{2}\mathequivrel{\Gamma}u_{1}$.
  By the induction hypothesis,
  there exists a finite sequence of terms $\mleft( t_{i} \mright)_{0\leq i \leq n}$ such that 
  $t_{0}\equiv u_{2}$, $t_{n}\equiv u_{1}$ and ${t_{i}=t_{i+1}}\in\mleft[\Gamma\mright]$ with $0\leq i < n$.
  Let $t'_{i}\equiv t_{n-i}$.
  The finite sequence of terms $\mleft( t'_{i} \mright)_{0\leq i \leq n}$ 
  satisfies $t'_{0}\equiv u_{1}$, $t'_{n}\equiv u_{2}$ and ${t'_{i}=t'_{i+1}}\in\mleft[\Gamma\mright]$.
  Thus, we have \ref{item:lemma-strcl}.

  Case 4.
  We consider the case where $u_{1}\mathequivrel{\Gamma}u_{3}$, $u_{3}\mathequivrel{\Gamma}u_{2}$.
  By the induction hypothesis,
  there exist two finite sequences of terms 
  $\mleft( t_{i} \mright)_{0\leq i \leq n}$, $\mleft( t'_{j} \mright)_{0\leq j \leq m}$ 
  such that $t_{0}\equiv u_{1}$, $t_{n}\equiv t'_{0}\equiv u_{3}$, $t'_{m}\equiv u_{2}$,
  ${t_{i}=t_{i+1}}\in\mleft[\Gamma\mright]$ and ${t'_{j}=t'_{j+1}}\in\mleft[\Gamma\mright]$
  with $0\leq i < n$, $0\leq j < m$.
  Define $\hat{t}_{k}$ as $t_{k}$ if $0\leq k < n$ and $t'_{k-n}$ if $n\leq k \leq n+m$.
  The finite sequence of terms 
  $\mleft( \hat{t}_{k} \mright)_{0\leq k \leq n}$ 
  satisfies
  $\hat{t}_{0}\equiv u_{1}$, $\hat{t}_{n}\equiv u_{2}$ and
  ${\hat{t}_{k}=\hat{t}_{k+1}}\in\mleft[\Gamma\mright]$.
  Thus, we have \ref{item:lemma-strcl}.

  Case 5. 
  We consider the case where 
  $\hat{u}_{1} \mathequivrel{\Gamma} \hat{u}_{2}$,
  $u_{1}\equiv u\mleft[v:=\hat{u}_{1}\mright]$ and $u_{2}\equiv u\mleft[v:=\hat{u}_{2}\mright]$.
  By the induction hypothesis,
  there exists a finite sequence of terms $\mleft( t_{i} \mright)_{0\leq i \leq n}$
  with $n\in\mathnat$ such that $t_{0}\equiv \hat{u}_{1}$, $t_{n}\equiv \hat{u}_{2}$,
  ${t_{i}=t_{i+1}}\in\mleft[\Gamma\mright]$ with $0\leq i < n$.

  Assume $v$ does not occur in $u$.
  In this case,  we have
  $u_{1}\equiv u\mleft[v:=\hat{u}_{1}\mright] \equiv u \equiv u\mleft[v:=\hat{u}_{2}\mright] \equiv u_{2}$.
  Hence, \ref{item:lemma-strcl} holds.

  Assume $v$ occurs in $u$.
  In this case, we have $u\equiv s^{m}v$ for some natural numbers $m$.
  Let $t'_{i}=s^{m}t_{i}$ for $0\leq i\leq n$.
  The finite sequence of terms  $\mleft( t'_{i} \mright)_{0\leq i \leq n}$ 
  satisfies $t'_{0}\equiv u_{1}$, $t'_{n}\equiv u_{2}$ and ${t'_{i}=t'_{i+1}}\in\mleft[\Gamma\mright]$.

  \ref{item:lemma-strcl} $\Rightarrow$ \ref{item:lemma-strcl-equivrel}:
  Assume \ref{item:lemma-strcl} to show \ref{item:lemma-strcl-equivrel}.
  By the assumption, there exists a finite sequence of terms $\mleft( t_{i} \mright)_{0\leq i \leq n}$ 
  with $n\in\mathnat$ such that
  $t_{0}\equiv u_{1}$, $t_{n}\equiv u_{2}$ and ${t_{i}=t_{i+1}}\in\mleft[\Gamma\mright]$ with  $0\leq i < n$.
  If ${t_{i}=t_{i+1}}\in\mleft[\Gamma\mright]$, then $t_{i}=t_{i+1}$ is $s^{n}\hat{t}_{1}=s^{n}\hat{t}_{2}$, 
  where $\hat{t}_{1}=\hat{t}_{2}\in \Gamma$ or $\hat{t}_{2}=\hat{t}_{1}\in \Gamma$.
  Therefore, $t_{i}\mathequivrel{\Gamma}t_{i+1}$.
  Because of the transitivity of $\mathequivrel{\Gamma}$, we have $u_{1}\mathequivrel{\Gamma}u_{2}$.
 \end{proof}

 \begin{lem}%
  \label{lemma:right_asp}
  For a set of formulas $\Gamma_{1}$ and
  $\Gamma_{2} \equiv \mleft(\Gamma_{1}, u_{1}=u'_{1}, u_{2}=u'_{2}, u_{3}=u'_{3}\mright)$,
  if $\mathvars{u'_{i}}$ ($i=1, 2, 3$) do not occur in $\Gamma_{1}, u_{1}, u_{2}, u_{3}, t, t'$
  and are all distinct variables, 
  %and $\mathvars{t}\neq \mathvars{u'_{i}}$ for $i=1, 2, 3$,
  then $t \mathequivrel{\Gamma_{2}} t'$  implies $t \mathequivrel{\Gamma_{1}} t'$.
 \end{lem}

 \begin{proof}
  Let $\mathvars{u'_{i}}=v_{i}$ for each $i=1$, $2$, $3$.
  Assume $t \mathequivrel{\Gamma_{2}} t'$,
  $t \not\mathdeprel{\Gamma_{1}} v_{i}$ for all $i=1, 2, 3$.
  By Lemma \ref{lemma:strcl}, 
  there exists a sequence $\mleft( t_{j} \mright)_{0\leq j \leq n}$ with $n\in\mathnat$
  such that $t_{0}\equiv t$, $t_{n}\equiv t'$ and
  ${t_{j}=t_{j+1}}\in\mleft[\Gamma_{2}\mright]$  with $0\leq j < n$.
  We show $t \mathequivrel{\Gamma_{1}} t'$  by induction on $n$.

  For $n=0$, we have $t \mathequivrel{\Gamma_{1}} t'$ immediately.

  We consider the case where $n>0$.

  If $t_{j}\not\equiv s^{m}u'_{i}$  for all $i=1, 2, 3$, $0\leq j \leq n$ and $m\in\mathnat$, 
  then  ${t_{j}=t_{j+1}}\in\mleft[\Gamma_{1}\mright]$ with $0\leq i < n$.
  By Lemma \ref{lemma:strcl},  we have $t \mathequivrel{\Gamma_{1}} t'$.
  
  Assume that there exists $j_{0}$ with $0\leq j_{0} \leq n$,
  such that $t_{j_{0}}\equiv s^{m}u'_{i}$ for some $i=1, 2, 3$ and $m\in\mathnat$.
  Since any formula of $\mleft[\Gamma_{2}\mright]$ in which $u'_{i}$ occurs
  is either $s^{l}u_{i}=s^{l}u'_{i}$ or $s^{l}u'_{i}=s^{l}u_{i}$ with $l\in\mathnat$
  and $\mathvars{u'_{i}}$ ($i=1, 2, 3$) do not occur in $t, t'$,
  we have $t_{j_{0}-1}\equiv t_{j_{0}+1} \equiv s^{m}u_{i}$.
  Define $\bar{t}_{k}$ as $t_{k}$ if $0\leq k < j_{0}$ and $t_{k+1}$ if $j_{0}\leq k \leq n-1$.
  Then, $\bar{t}_{0}\equiv t$, $\bar{t}_{n-1}\equiv t'$ and
  ${\bar{t}_{k}=\bar{t}_{k+1}}\in\mleft[\Gamma_{2}\mright]$ with $0\leq k < n-1$.
  By the induction hypothesis, we have $t \mathequivrel{\Gamma_{1}} t'$. 
 \end{proof}

 \begin{lem}%
  \label{lemma:left_asp_1}
  For a set of formulas $\Gamma_{1}$ and 
  $\Gamma_{2} \equiv \mleft(\Gamma_{1}, u_{1}=u'_{1}, \dots, u_{n}=u'_{n}\mright)$ with a natural number $n$,
  if $t \not\mathdeprel{\Gamma_{1}} u_{i}$ and $t \not\mathdeprel{\Gamma_{1}} u'_{i}$ with $i=1, \dots, n $,
  then $t \mathequivrel{\Gamma_{2}} t'$ implies $t \mathequivrel{\Gamma_{1}} t'$.
 \end{lem}

 \begin{proof}
  Assume $t \not\mathdeprel{\Gamma_{1}} u_{i}$, $t \not\mathdeprel{\Gamma_{1}} u'_{i}$ for $i=1, \dots, n$,
  and $t \mathequivrel{\Gamma_{2}} t'$.
  By Lemma \ref{lemma:strcl}, 
  there exists a sequence $\mleft( t_{j} \mright)_{0\leq j \leq m}$ with $m\in\mathnat$ such that
  $t_{0}\equiv t$, $t_{m}\equiv t'$ and ${t_{j}=t_{j+1}}\in\mleft[\Gamma_{2}\mright]$ 
  with $0\leq j < m$.
  
  If $t_{j}\not\equiv s^{l}u_{i}$ and $t_{j}\not\equiv s^{l}u'_{i}$ 
  for all $0\leq j \leq n$, $i=1, \dots, n$, and any $l\in\mathnat$,
  then ${t_{j}=t_{j+1}}\in\mleft[\Gamma_{1}\mright]$ with all $0\leq j < m$.
  By Lemma \ref{lemma:strcl}, we have $t \mathequivrel{\Gamma_{1}} t'$.
  
  Assume that there exists $j$ with $0\leq j \leq n$,
  such that $t_{j}\equiv s^{l}u_{i}$ or $t_{j}\equiv s^{l}u'_{i}$ for $i=1, \dots, n$,
  and some $l\in\mathnat$.
  Let $j_{0}$ be the least number among such $j$'s.
  Since $j_{0}$ is the least,
  we have ${t_{j}=t_{j+1}}\in\mleft[\Gamma_{1}\mright]$ for all $0\leq j < j_{0}$.
  By Lemma \ref{lemma:strcl},
  we have $t\mathequivrel{\Gamma_{1}} s^{l}u_{i}$ or $t\mathequivrel{\Gamma_{1}} s^{l}u'_{i}$.
  This contradicts $t \not\mathdeprel{\Gamma_{1}} u_{i}$ and $t \not\mathdeprel{\Gamma_{1}} u'_{i}$.
 \end{proof}

  We call the assumption of (\rulename{Case $\mathindAddtwosy$}) whose form 
 \[\Gamma, a=sx, b=y, c=z, \mathindAddtwo{x}{sy}{z} \fCenter \Delta \]
 \emph{the right assumption of} the rule. 
 The other assumption is called \emph{the left assumption of} the rule. 
 
 \begin{lem}%
  \label{lemma:abc_relations}
  Let $\Gamma \fCenter \Delta$ be in $\mathdertreecf$ and
  \begin{align*}
   \mathsetA{\Gamma \fCenter \Delta}&=
   \mathsetintension{a}{
   \mathindAddtwo{a}{b}{c} \in \Gamma,
   \mathindAddone{a}{b}{c} \in \Delta,
   \text{ or } a\equiv 0} \text{ and } \\
   \mathsetBC{\Gamma \fCenter \Delta}&= 
   \mathsetintension{b}{ 
   \mathindAddtwo{a}{b}{c} \in \Gamma \text{ or }
   \mathindAddone{a}{b}{c} \in \Delta
   } \cup
   \mathsetintension{c}{ 
   \mathindAddtwo{a}{b}{c} \in \Gamma \text{ or }
   \mathindAddone{a}{b}{c} \in \Delta
   }.
  \end{align*}
  
  If $t\in\mathsetA{\Gamma \fCenter \Delta}$ and
  $u\in\mathsetBC{\Gamma \fCenter \Delta}$,
  then $t\not\mathdeprel{\Gamma}u$.
 \end{lem}

 \begin{proof}
  We prove the statement 
  by induction on the height of the node $\Gamma\fCenter\Delta$ in $\mathdertreecf$. 

  The root of $\mathdertreecf$ satisfies the statement.

  Assume $\Gamma \fCenter \Delta$ is not the root.
  Let $\Gamma' \fCenter \Delta'$ be the parent of $\Gamma \fCenter \Delta$.
  We consider cases according to the rule with the conclusion $\Gamma' \fCenter \Delta'$.

  Case 1. In the case (\rulename{Weak}),
  we have the statement by $\Gamma\subseteq\Gamma'$.

  %%%%%%%%%%%%%%%%%%%%%%%%%%%%%%%%%%%

  Case 2. In the case (\rulename{Subst}),
  we have the statement by  Lemma \ref{lemma:subst_rel} \ref{item:lemma-subst_rel-deprel}.

  %%%%%%%%%%%%%%%%%%%%%%%%%%%%%%%%%%%

  Case 3. In the case (\rulenameLa),
  we have the statement by Lemma \ref{lemma:eq_rel} \ref{item:lemma-eq_rel-deprel}.
  
  %%%%%%%%%%%%%%%%%%%%%%%%%%%%%%%%%%%
  
  Case 4. 
  We consider the case where the rule is (\rulename{Case $\mathindAddtwosy$})
  and $\Gamma\fCenter\Delta$ is the right assumption of the rule.
  Let $\mathindAddtwo{a}{b}{c}$ be the principal formula of the rule. 
  There exists $\Pi$ such that 
  $\Gamma' \equiv \mleft(\Pi, \mathindAddtwo{a}{b}{c}\mright)$ and 
  $\Gamma \equiv \mleft(\Pi, a = sx, b = y, c = z, \mathindAddtwo{x}{sy}{z}\mright)$ 
  for fresh variables $x$, $y$, $z$.

  Assume $t\in\mathsetA{\Gamma \fCenter \Delta}$ and $u\in\mathsetBC{\Gamma \fCenter \Delta}$
  and $t\mathdeprel{\Gamma}u$ for contradiction.

  Define $\hat{t}$ as $a$ if $t\equiv x$ and $t$ otherwise.
  We also define $\hat{u}$ as $b$ if $u\equiv sy$, $c$ if $u\equiv z$ and $u$ otherwise.
  Since $t\mathdeprel{\Gamma}u$ holds,
  we have $\hat{t}\mathdeprel{\Gamma}\hat{u}$.
  By Lemma \ref{lemma:right_asp},
  we have $\hat{t}\mathdeprel{\Gamma'}\hat{u}$.
  Since $\hat{t}\in\mathsetA{\Gamma' \fCenter \Delta'}$ and
  $\hat{u}\in\mathsetBC{\Gamma' \fCenter \Delta'}$ hold,
  this contradicts the induction hypothesis.

  %%%%%%%%%%%%%%%%%%%%%%%%%%%%%%%%%%%

  Case 5. 
  We consider the case where the rule is (\rulename{Case $\mathindAddtwosy$})
  and $\Gamma\fCenter\Delta$ is the left assumption of the rule.

  Let $\mathindAddtwo{a}{b}{c}$
  be the principal formula of the rule.
  There exists $\Pi$ such that 
  $\Gamma' \equiv \mleft(\Pi, \mathindAddtwo{a}{b}{c}\mright)$ and
  $\Gamma \equiv \mleft(\Pi, a = 0, b = y, c = y\mright)$ 
  for a fresh variable $y$.
  Let $\Pi'\equiv \mleft(\Pi, b = y, c = y\mright)$.

  Let $t\in\mathsetA{\Gamma \fCenter \Delta}$ and $u\in\mathsetBC{\Gamma \fCenter \Delta}$.
  Since $\mathsetA{\Gamma \fCenter \Delta}\subseteq\mathsetA{\Gamma' \fCenter \Delta'}$ holds,
  we have $t\in\mathsetA{\Gamma' \fCenter \Delta'}$.
  By $\mathsetBC{\Gamma \fCenter \Delta}\subseteq\mathsetBC{\Gamma' \fCenter \Delta'}$,
  we have $u\in\mathsetBC{\Gamma' \fCenter \Delta'}$.
  By the induction hypothesis, 
  $t\not\mathdeprel{\Gamma'}u$, $t\not\mathdeprel{\Gamma'}b$ and $t\not\mathdeprel{\Gamma'}c$.
  Since the set of formulas with ${=}$ in $\Pi$ is the same as the set of formulas with ${=}$ in $\Gamma'$, 
  we have $t\not\mathdeprel{\Pi}u$, $t\not\mathdeprel{\Pi}b$ and $t\not\mathdeprel{\Pi}c$.
  By Lemma \ref{lemma:left_asp_1}, $t\not\mathdeprel{\Pi'}u$.

  By the induction hypothesis,  $u\not\mathdeprel{\Gamma'}a$, $a\not\mathdeprel{\Gamma'}b$
  and $a\not\mathdeprel{\Gamma'}c$.
  Since the set of formulas with ${=}$ in $\Pi$ is the same as the set of formulas with ${=}$ in $\Gamma'$, 
  $u\not\mathdeprel{\Pi}a$, $a\not\mathdeprel{\Pi}b$
  and $a\not\mathdeprel{\Pi}c$.
  By Lemma \ref{lemma:left_asp_1}, $u\not\mathdeprel{\Pi'}a$.

  By the induction hypothesis,  
  $u\not\mathdeprel{\Gamma'}0$, $0\not\mathdeprel{\Gamma'}b$ and $0\not\mathdeprel{\Gamma'}c$.
  Since the set of formulas with ${=}$ in $\Pi$ is the same as the set of formulas with ${=}$ in $\Gamma'$, 
  $u\not\mathdeprel{\Pi}0$, $0\not\mathdeprel{\Pi}b$ and $0\not\mathdeprel{\Pi}c$.
  By Lemma \ref{lemma:left_asp_1},  $u\not\mathdeprel{\Pi'}0$.

  By Lemma \ref{lemma:left_asp_1} and these three facts, $t\not\mathdeprel{\Gamma}u$.

  %%%%%%%%%%%%%%%%%%%%%%%%%%%%%%%%%%%

  Case 6. 
  In the case (\rulename{$\mathindAddonesy$ R${}_2$}), $\Gamma\equiv \Gamma'$ implies the statement
  by the induction hypothesis.
 \end{proof}

 \subsection{Index} \label{subsec:Index}
 In this section, we define a key concept, called an index, to prove Theorem \ref{thm:main}.

  \begin{definition}[Index]%
  \label{def:index}
  For a sequent $\Gamma\fCenter\Delta$ and $\mathindAddtwo{a}{b}{c}\in\Gamma$,
  we define \emph{the index of  $\mathindAddtwo{a}{b}{c}$ in $\Gamma\fCenter\Delta$} as follows:
  \begin{enumerate}
   \item If $b \not\mathdeprel{\Gamma} b'$  for any $\mathindAddone{a'}{b'}{c'} \in \Delta$,
	 then the index of $\mathindAddtwo{a}{b}{c}$ in $\Gamma\fCenter\Delta$ is $\bot$, and
   \item if there exists uniquely $m-n$
	 such that $n$, $m \in \mathnat$,  $s^{n} b \mathequivrel{\Gamma} s^{m} b'$ 
	 and $\mathindAddone{a'}{b'}{c'} \in \Delta$,
	 then the index of $\mathindAddtwo{a}{b}{c}$ in $\Gamma\fCenter\Delta$ is $m-n$
	 (namely the uniqueness means that $s^{n'} b \mathequivrel{\Gamma} s^{m'} b''$
	 for $m'$,$n'\in\mathnat$ and $\mathindAddone{a''}{b''}{c''} \in \Delta$ imply $m-n=m'-n'$).
  \end{enumerate}
 \end{definition}

 Note that
 if there exists $n, m \in \mathnat$ such that $s^{n} b \mathequivrel{\Gamma} s^{m} b'$  
 for some $\mathindAddone{a'}{b'}{c'}$ and $m-n$ is not unique,
 then the index of $\mathindAddtwo{a}{b}{c}$ in $\Gamma\fCenter\Delta$ is undefined.

 \begin{definition}[Index sequent]%
  \label{def:annoying}%  
  The sequent $\Gamma \fCenter \Delta$ is said to be
  an \emph{index sequent} if the following conditions hold:
  \begin{enumerate}
   \item If $t\in\mathsetBone{\Gamma \fCenter \Delta}$ and $u\in\mathsetC{\Gamma \fCenter \Delta}$,
	 then $t\not\mathdeprel{\Gamma}u$, and
	 \label{item:def-annoying_bc}
   \item if $s^{n}b \mathequivrel{\Gamma} s^{m}b'$ with $b$, $b'\in \mathsetBone{\Gamma \fCenter \Delta}$,
	 then $n=m$, where
	 \label{item:def-annoying_tsnt}
  \end{enumerate}
  \begin{align*}
   \mathsetBone{\Gamma \fCenter \Delta}&=
   \mathsetintension{b}{
   \mathindAddone{a}{b}{c} \in \Delta
   }, \text{ and } \\
   \mathsetC{\Gamma \fCenter \Delta}&= 
   \mathsetintension{c}{ 
   \mathindAddtwo{a}{b}{c} \in \Gamma \text{ or }
   \mathindAddone{a}{b}{c} \in \Delta
   }.
  \end{align*}

 \end{definition}

 This condition \ref{item:def-annoying_tsnt}
 guarantees the existence of an index, as shown in the following lemma. 
 We will use \ref{item:def-annoying_bc} to calculate an index 
 in Lemma \ref{lemma:index} \ref{item:lemma-index_neg} and an infinite sequence 
 in Lemma \ref{lemma:rightmost}.

 \begin{lem}
  If $\Gamma\fCenter\Delta$ is an index sequent, 
  the index of any $\mathindAddtwo{a}{b}{c}\in\Gamma$ in $\Gamma\fCenter\Delta$ is defined.
 \end{lem}

 \begin{proof}
  If $b\not\mathdeprel{\Gamma}b'$ for any $\mathindAddone{a'}{b'}{c'} \in \Delta$,
  then the index is $\bot$.

  Assume $b\mathdeprel{\Gamma}b'_{0}$ for some $\mathindAddone{a'_{0}}{b'_{0}}{c'_{0}}\in\Delta$.
  By Definition \ref{def:deprel},
  there exist $n_{0}$ and $m_{0}$ such that $s^{n_{0}} b \mathequivrel{\Gamma} s^{m_{0}} b'_{0}$.
  To show the uniqueness,
  we fix $\mathindAddone{a'_{1}}{b'_{1}}{c'_{1}} \in \Delta$
  and assume $s^{n_{1}} b \mathequivrel{\Gamma} s^{m_{1}} b'_{1}$.
  Since $s^{n_{0}+n_{1}} b \mathequivrel{\Gamma} s^{m_{0}+n_{1}} b'_{0}$
  and $s^{n_{1}+n_{0}} b \mathequivrel{\Gamma} s^{m_{1}+n_{0}} b'_{1}$,
  we have $s^{m_{0}+n_{1}} b'_{0} \mathequivrel{\Gamma} s^{m_{1}+n_{0}} b'_{1}$.
  From \ref{item:def-annoying_tsnt} of Definition \nolinebreak \ref{def:annoying}, $m_{0}+n_{1}=m_{1}+n_{0}$.
  Thus, $m_{0}-n_{0}=m_{1}-n_{1}$.
 \end{proof}

 \begin{definition}[Switching point]%
  \label{def:bad_app}
  A node $\sigma$ in a derivation tree is called a \emph{switching point}
  if the rule with the conclusion $\sigma$ is (\rulename{Case $\mathindAddtwosy$}) and
  the index of the principal formula for the rule in the conclusion is $\bot$.
 \end{definition}

 \begin{definition}[Index path] 
  \label{def:bad-path}       
  A path $\mleft( \Gamma_{i} \fCenter \Delta_{i} \mright)_{0\leq i <\alpha}$ 
  in $\mathtreeunfolding{\mathproofcf}$
  with some $\alpha \in \mathnat \cup \mathsetextension{\omega}$
  is said to be an \emph{index path}
  if the following conditions hold:
  \begin{enumerate}
   \item $\Gamma_{0}\fCenter\Delta_{0}$ is an index sequent, and
   \item if the rule for $\Gamma_{i} \fCenter \Delta_{i}$ is $\text{(\rulename{Case $\mathindAddtwosy$})}$
	 and $\Gamma_{i+1}\fCenter\Delta_{i+1}$ is the left assumption of the rule,
	 then $\Gamma_{i}\fCenter\Delta_{i}$ is a switching point.
  \end{enumerate}
 \end{definition}

 \begin{lem}%
  \label{lemma:annoying}
  Every sequent in an index path is an index sequent.
 \end{lem}

 \begin{proof}
  Let $\mleft( \Gamma_{i} \fCenter \Delta_{i} \mright)_{0\leq i <\alpha}$ be an index path. 
  We use $\mathsetBone{\Gamma \fCenter \Delta}$ and $\mathsetC{\Gamma \fCenter \Delta}$ in
  Definition \ref{def:annoying}.
  We prove the statement by the induction on $i$.
  
  For $i=0$, $\Gamma_{0} \fCenter \Delta_{0}$ is an index sequent by Definition \ref{def:bad-path}.

  For $i>0$, we consider cases according to the rule with the conclusion $\Gamma_{i-1}\fCenter\Delta_{i-1}$.

  %%%%%%%%%%%%%%%%%%%%%%%%%%%%%%%%%%%%%%%%%

  Case 1. The case (\rulename{Weak}).

  \ref{item:def-annoying_bc} 
  Assume that
  $t\in\mathsetBone{\Gamma_{i} \fCenter \Delta_{i}}$ and $u\in\mathsetC{\Gamma_{i} \fCenter \Delta_{i}}$.
  Since 
  $\mathsetBone{\Gamma_{i} \fCenter \Delta_{i}}\subseteq\mathsetBone{\Gamma_{i-1} \fCenter \Delta_{i-1}}$
  holds, we have $t\in\mathsetBone{\Gamma_{i-1} \fCenter \Delta_{i-1}}$.
  By $\mathsetC{\Gamma_{i} \fCenter \Delta_{i}}\subseteq\mathsetC{\Gamma_{i-1} \fCenter \Delta_{i-1}}$,
  we have $u\in\mathsetC{\Gamma_{i-1} \fCenter \Delta_{i-1}}$.
  By the induction hypothesis  \ref{item:def-annoying_bc}, we have $t\not\mathdeprel{\Gamma_{i-1}}u$.
  By $\Gamma_{i}\subseteq\Gamma_{i-1}$, we have $t\not\mathdeprel{\Gamma_{i}}u$.

  \ref{item:def-annoying_tsnt} 
  Assume that $s^{n}b \mathequivrel{\Gamma_{i}} s^{m}b'$ 
  with $b$, $b'\in \mathsetBone{\Gamma_{i} \fCenter \Delta_{i}}$ for $n$, $m\in\mathnat$.
  By $\Gamma_{i}\subseteq\Gamma_{i-1}$, we have $s^{n}b \mathequivrel{\Gamma_{i-1}} s^{m}b'$.
  Since $\mathsetBone{\Gamma_{i}\fCenter\Delta_{i}}\subseteq\mathsetBone{\Gamma_{i-1}\fCenter\Delta_{i-1}}$
  holds, we have $b$, $b'\in \mathsetBone{\Gamma_{i-1} \fCenter \Delta_{i-1}}$.
  By the induction hypothesis  \ref{item:def-annoying_tsnt} , we have $n=m$.

  %%%%%%%%%%%%%%%%%%%%%%%%%%%%%%%%%%%%%%%%%%%%%

  Case 2. The case (\rulename{Subst}) with a substitution $\theta$.
  
  \ref{item:def-annoying_bc} 
  Assume that 
  $t\in\mathsetBone{\Gamma_{i} \fCenter \Delta_{i}}$ and $u\in\mathsetC{\Gamma_{i} \fCenter \Delta_{i}}$.
  Since $\Gamma_{i-1}\equiv \Gamma_{i}\mleft[\theta\mright]$ 
  and $\Delta_{i-1}\equiv \Delta_{i}\mleft[\theta\mright]$ hold,
  we have $t\mleft[\theta\mright]\in\mathsetBone{\Gamma_{i-1} \fCenter \Delta_{i-1}}$ and
  $u\mleft[\theta\mright]\in\mathsetC{\Gamma_{i-1} \fCenter \Delta_{i-1}}$.
  By the induction hypothesis  \ref{item:def-annoying_bc} , 
  we have $t\mleft[\theta\mright]\not\mathdeprel{\Gamma_{i-1}}u\mleft[\theta\mright]$.
  By Lemma \ref{lemma:subst_rel} \ref{item:lemma-subst_rel-deprel}, 
  we have $t\not\mathdeprel{\Gamma_{i}}u$.

  \ref{item:def-annoying_tsnt} 
  Assume that
  $s^{n}b \mathequivrel{\Gamma_{i}} s^{m}b'$ with $b$, $b' \in \mathsetBone{\Gamma_{i} \fCenter \Delta_{i}}$
  for $n$, $m\in\mathnat$.
  By Lemma \ref{lemma:subst_rel} \ref{item:lemma-subst_rel-equiv},
  $s^{n}b\mleft[\theta\mright] \mathequivrel{\Gamma_{i-1}} s^{m}b'\mleft[\theta\mright]$.
  Since $\Delta_{i-1}\equiv \Delta_{i}\mleft[\theta\mright]$ holds,
  we have 
  $b\mleft[\theta\mright]$, $b'\mleft[\theta\mright]\in \mathsetBone{\Gamma_{i-1} \fCenter \Delta_{i-1}}$.
  By the induction hypothesis   \ref{item:def-annoying_tsnt}, we have $n=m$.
  
  %%%%%%%%%%%%%%%%%%%%%%%%%%%%%%%%%%%%%%%%%%

  Case 3. The case (\rulenameLa).

  Let $u_{1}=u_{2}$ be the principal formula of the rule.
  There exist $\Gamma$ and $\Delta$ such that
 \begin{align*}
  \Gamma_{i-1}&\equiv {\mleft(\Gamma\mleft[v_{1}:=u_{1}, v_{2}:=u_{2}\mright], {u_{1} = u_{2}} \mright)}, \\
  \Delta_{i-1}&\equiv {\mleft(\Delta\mleft[v_{1}:=u_{1}, v_{2}:=u_{2}\mright], {u_{1} = u_{2}} \mright)}, \\
  \Gamma_{i}&\equiv {\mleft(\Gamma\mleft[v_{1}:=u_{2}, v_{2}:=u_{1}\mright], {u_{1} = u_{2}}\mright)}, 
  \text{ and } \\ 
  \Delta_{i}&\equiv {\mleft(\Delta\mleft[v_{1}:=u_{2}, v_{2}:=u_{1}\mright], {u_{1} = u_{2}} \mright)}.
 \end{align*}
  
  \ref{item:def-annoying_bc} 
  Assume that
  $t\in\mathsetBone{\Gamma_{i} \fCenter \Delta_{i}}$ and 
  $u\in\mathsetC{\Gamma_{i} \fCenter \Delta_{i}}$.
  From the definition of $\Gamma_{i}$ and $\Delta_{i}$,
  there exist terms $\hat{t}$, $\hat{u}$ such that
  $t\equiv \hat{t}\mleft[v_{1}:=u_{2}, v_{2}:=u_{1}\mright]$ and
  $u\equiv \hat{u}\mleft[v_{1}:=u_{2}, v_{2}:=u_{1}\mright]$.
  Then,
  $\hat{t}\mleft[v_{1}:=u_{1}, v_{2}:=u_{2}\mright]\in\mathsetBone{\Gamma_{i-1} \fCenter \Delta_{i-1}}$ and
  $\hat{u}\mleft[v_{1}:=u_{1}, v_{2}:=u_{2}\mright]\in\mathsetC{\Gamma_{i-1} \fCenter \Delta_{i-1}}$.
  By the induction hypothesis  \ref{item:def-annoying_bc}, 
  we have $\hat{t}\mleft[v_{1}:=u_{1}, v_{2}:=u_{2}\mright]\not\mathdeprel{\Gamma_{i-1}}\hat{u}\mleft[v_{1}:=u_{1}, v_{2}:=u_{2}\mright]$.
  By Lemma \ref{lemma:eq_rel} \ref{item:lemma-eq_rel-deprel},
  we have 
  $\hat{t}\mleft[v_{1}:=u_{2}, v_{2}:=u_{1}\mright]\not\mathdeprel{\Gamma_{i}}\hat{u}\mleft[v_{1}:=u_{2}, v_{2}:=u_{1}\mright]$. 
  Thus, $t\not\mathdeprel{\Gamma_{i}}u$.

  \ref{item:def-annoying_tsnt} 
  Assume that
  $s^{n}b \mathequivrel{\Gamma_{i}} s^{m}b'$
  with $b$, $b'\in \mathsetBone{\Gamma_{i} \fCenter \Delta_{i}}$ for $n$, $m\in\mathnat$. 
  From the definition of $\Gamma_{i}$ and $\Delta_{i}$,
  there exist terms $\hat{b}$, $\hat{b}'\in\Delta$ such that
  $b\equiv s^{n}\hat{b}\mleft[v_{1}:=u_{2}, v_{2}:=u_{1}\mright]$ and
  $b'\equiv s^{m}\hat{b}'\mleft[v_{1}:=u_{2}, v_{2}:=u_{1}\mright]$.
  By Lemma \ref{lemma:eq_rel} \ref{item:lemma-eq_rel-equiv},
  $s^{n}\hat{b}\mleft[v_{1}:=u_{1}, v_{2}:=u_{2}\mright] \mathequivrel{\Gamma_{i-1}} s^{m}\hat{b}'\mleft[v_{1}:=u_{1}, v_{2}:=u_{2}\mright]$.
  From the definition of $\Gamma_{i-1}$ and $\Delta_{i-1}$, 
  $\hat{b}\mleft[v_{1}:=u_{1}, v_{2}:=u_{2}\mright]$, 
  $\hat{b}'\mleft[v_{1}:=u_{1}, v_{2}:=u_{2}\mright]\in \mathsetBone{\Gamma_{i-1} \fCenter \Delta_{i-1}}$.
  By the induction hypothesis \ref{item:def-annoying_tsnt}, we have $n=m$.

  %%%%%%%%%%%%%%%%%%%%%%%%%%%%%
  
  Case 4. 
  The case (\rulename{Case $\mathindAddtwosy$})
  with the right assumption $\Gamma_{i}\fCenter\Delta_{i}$.
  
  Let $\mathindAddtwo{a}{\hat{b}}{c}$ be the principal formula of the rule.
  There exists $\Pi$ such that 
  $\Gamma_{i-1}\equiv \mleft(\Pi, \mathindAddtwo{a}{\hat{b}}{c}\mright)$ and
  $\Gamma_{i}\equiv \mleft(\Pi, a = sx, \hat{b} = y, c = z, \mathindAddtwo{x}{sy}{z}\mright)$
  for fresh variables $x$, $y$, $z$.
  
  \ref{item:def-annoying_bc}
  Assume that
  $t\in\mathsetBone{\Gamma_{i} \fCenter \Delta_{i}}$ and $u\in\mathsetC{\Gamma_{i} \fCenter \Delta_{i}}$.
  Assume that $t\mathdeprel{\Gamma_{i}}u$ for contradiction.
  Define $\hat{u}$ as $c$ if $u\equiv z$ and $u$ otherwise.
  Since $t\mathdeprel{\Gamma_{i}}u$ holds, we have $t\mathdeprel{\Gamma_{i}}\hat{u}$.
  By Lemma \ref{lemma:right_asp}, 
  we have $t\mathdeprel{\Gamma_{i-1}}\hat{u}$.
  Since $t\in\mathsetBone{\Gamma_{i-1} \fCenter \Delta_{i-1}}$ and
  $\hat{u}\in\mathsetC{\Gamma_{i-1} \fCenter \Delta_{i-1}}$ hold,
  this contradicts the induction hypothesis \ref{item:def-annoying_bc}.

  \ref{item:def-annoying_tsnt} 
  Assume that $s^{n}b \mathequivrel{\Gamma_{i}} s^{m}b'$ with $b$, 
  $b'\in \mathsetBone{\Gamma_{i} \fCenter \Delta_{i}}$ for $n$, $m\in\mathnat$.
  By Lemma \ref{lemma:right_asp},  $s^{n}b \mathequivrel{\Gamma_{i-1}} s^{m}b'$.
  Since $\Delta_{i-1} \equiv \Delta_{i}$ holds,
  we have  $b$, $b'\in \mathsetBone{\Gamma_{i-1} \fCenter \Delta_{i-1}}$.
  By the induction hypothesis \ref{item:def-annoying_tsnt}, we have $n=m$.

  %%%%%%%%%%%%%%%%%%%%%%%%%%%%%%%%%%%%%%%%

  Case 5. 
  The case (\rulename{Case $\mathindAddtwosy$})
  with the left assumption $\Gamma_{i} \fCenter \Delta_{i}$.
  In this case, $\Gamma_{i-1} \fCenter \Delta_{i-1}$ is a switching point.

  Let $\mathindAddtwo{a}{\hat{b}}{c}$ be the principal formula of the rule. 
  There exists $\Pi$ such that 
  $\Gamma_{i-1} \equiv {\mleft(\Pi, \mathindAddtwo{a}{\hat{b}}{c}\mright)}$ and 
  $\Gamma_{i} \equiv {\mleft(\Pi, a = 0, \hat{b} = y, c = y\mright)}$ with a fresh variable $y$. 
  
  \ref{item:def-annoying_bc} 
  Assume that $t\in\mathsetBone{\Gamma_{i} \fCenter \Delta_{i}}$ and
  $u\in\mathsetC{\Gamma_{i} \fCenter \Delta_{i}}$.
  Since $\mathsetBone{\Gamma_{i} \fCenter \Delta_{i}}=\mathsetBone{\Gamma_{i-1} \fCenter \Delta_{i-1}}$ holds,
  we have $t\in\mathsetBone{\Gamma_{i-1} \fCenter \Delta_{i-1}}$.
  By $\mathsetC{\Gamma_{i} \fCenter \Delta_{i}}\subseteq\mathsetC{\Gamma_{i-1} \fCenter \Delta_{i-1}}$,
  we have $u\in\mathsetC{\Gamma_{i-1} \fCenter \Delta_{i-1}}$.
  By the induction hypothesis \ref{item:def-annoying_bc},
  $t\not\mathdeprel{\Gamma_{i-1}}u$ and $t\not\mathdeprel{\Gamma_{i-1}}c$.
  By Lemma \ref{lemma:abc_relations}, 
  $t\not\mathdeprel{\Gamma_{i-1}}a$ and $t\not\mathdeprel{\Gamma_{i-1}}0$.
  Since $y$ is fresh, we have $t\not\mathdeprel{\Gamma_{i-1}}y$.
  Since $\Gamma_{i-1} \fCenter \Delta_{i-1}$ is a switching point,
  we have $t\not\mathdeprel{\Gamma_{i-1}}\hat{b}$.
  By Lemma \ref{lemma:left_asp_1},  $t\not\mathdeprel{\Gamma_{i}}u$.

  \ref{item:def-annoying_tsnt}
  Assume that $s^{n}b \mathequivrel{\Gamma_{i}} s^{m}b'$ 
  with $b$, $b'\in \mathsetBone{\Gamma_{i} \fCenter \Delta_{i}}$ for $n$, $m\in\mathnat$ to show $n=m$.
  By Lemma \ref{lemma:abc_relations}, 
  $s^{n}b\not\mathdeprel{\Gamma_{i}}a$ and $s^{n}b\not\mathdeprel{\Gamma_{i}}0$.
  Since $\Gamma_{i-1} \fCenter \Delta_{i-1}$ is a switching point,
  we have $s^{n}b \not\mathdeprel{\Gamma_{i-1}} \hat{b}$.
  By the induction hypothesis \ref{item:def-annoying_bc}, $s^{n}b \not\mathdeprel{\Gamma_{i-1}} c$. 
  Since $y$ is fresh, we have $s^{n}b\not\mathdeprel{\Gamma_{i-1}}y$.
  By Lemma \ref{lemma:left_asp_1}, we have $s^{n}b \mathequivrel{\Gamma_{i-1}} s^{m}b'$.
  Because of $\mathsetBone{\Gamma_{i} \fCenter \Delta_{i}}=\mathsetBone{\Gamma_{i-1} \fCenter \Delta_{i-1}}$,
  we have $b$, $b'\in \mathsetBone{\Gamma_{i-1} \fCenter \Delta_{i-1}}$.
  By the induction hypothesis \ref{item:def-annoying_tsnt}, we have $n=m$.

  %%%%%%%%%%%%%%%%%%%%%%%%%%%%%%%%%%%%%%%%%

  Case 6. The case (\rulename{$\mathindAddonesy$ R${}_2$}).
  Let $\mathindAddone{sa}{\hat{b}}{sc}$ be the principal formula of the rule.

  \ref{item:def-annoying_bc}
  Assume that 
  $t\in\mathsetBone{\Gamma_{i} \fCenter \Delta_{i}}$ and $u\in\mathsetC{\Gamma_{i} \fCenter \Delta_{i}}$
  and $t\mathdeprel{\Gamma_{i}}u$ for contradiction.
  Define $\hat{u}$ as $sc$ if $u\equiv c$ and $u$ otherwise.
  Since $t\mathdeprel{\Gamma_{i}}u$ holds, we have $t\mathdeprel{\Gamma_{i}}\hat{u}$.
  Since $\Gamma_{i-1}=\Gamma_{i}$ holds, we have $t\mathdeprel{\Gamma_{i-1}}\hat{u}$.
  Since $t\in\mathsetBone{\Gamma_{i-1} \fCenter \Delta_{i-1}}$ and
  $\hat{u}\in\mathsetC{\Gamma_{i-1} \fCenter \Delta_{i-1}}$ hold,
  this contradicts the induction hypothesis \ref{item:def-annoying_bc}.

  \ref{item:def-annoying_tsnt}
  Assume that
  $s^{n}b \mathequivrel{\Gamma_{i}} s^{m}b'$ with $b$, $b'\in \mathsetBone{\Gamma_{i} \fCenter \Delta_{i}}$
  for $n$, $m\in\mathnat$.
  Because $\Gamma_{i-1}=\Gamma_{i}$,
  we have $s^{n}b \mathequivrel{\Gamma_{i-1}} s^{m}b'$.
  Since the second argument of a formula with $\mathindAddonesy$ in $\Delta_{i}$ is that in $\Delta_{i-1}$,
  we have  $b$, $b'\in \mathsetBone{\Gamma_{i-1} \fCenter \Delta_{i-1}}$.
  By the induction hypothesis \ref{item:def-annoying_tsnt}, we have $n=m$.
 \end{proof}

 \begin{lem}%
  \label{lemma:index}
  For an index path $\mleft(\Gamma_{i}\fCenter\Delta_{i}\mright)_{0\leq i< \alpha}$ and
  a trace $\mleft( \tau_{k} \mright)_{k \geq 0}$ following 
  $\mleft(\Gamma_{i}\fCenter\Delta_{i}\mright)_{i\geq p}$,
  if $d_{k}$ is the index of $\tau_{k}$, the following statements holds:
  
  \begin{enumerate}
   \item If $d_{k} =\bot$, then $d_{k+1}=\bot$. 
	 \label{item:lemma-index_neg}
   \item If the rule with the conclusion $\Gamma_{p+k}\fCenter\Delta_{p+k}$ is 
	 (\rulename{Weak}) or (\rulename{Subst}),
	 then $d_{k+1}=d_{k}$ or $d_{k+1}=\bot$.
	 \label{item:lemma-index_regress}
   \item If the rule with the conclusion $\Gamma_{p+k}\fCenter\Delta_{p+k}$ is 
	 (\rulenameLa) or (\rulename{\rulename{$\mathindAddonesy$ R${}_2$}}),
	 then $d_{k+1}=d_{k}$.
	 \label{item:lemma-index_not_change}
   \item Assume the rule with the conclusion $\Gamma_{p+k}\fCenter\Delta_{p+k}$ is 
	 (\rulename{Case $\mathindAddtwosy$}). 
	 \label{item:lemma-index_Case}
	 \begin{enumerate}
	  \item If $\Gamma_{p+k+1}\fCenter\Delta_{p+k+1}$ is the left assumption of the rule,
		then $d_{k+1}=d_{k}$.
		\label{item:lemma-index_left_ass}
	  \item If $\Gamma_{p+k+1}\fCenter\Delta_{p+k+1}$ is the right assumption of the rule and
		$\tau_{k}$ is not a progress point of the trace,
		then $d_{k+1}=d_{k}$.
		\label{item:lemma-index_not_progress}
	  \item If $\Gamma_{p+k+1}\fCenter\Delta_{p+k+1}$ is the right assumption of the rule and
		$\tau_{k}$ is a progress point of the trace,
		then $d_{k+1}=d_{k}+1$.
		\label{item:lemma-index_progress}
 	 \end{enumerate}
  \end{enumerate}
 \end{lem}

 \begin{proof}
  Let $\tau_{k} \equiv \mathindAddtwo{a_{k}}{b_{k}}{c_{k}}$.

  \noindent \ref{item:lemma-index_neg}
  It suffices to show that
  $b_{k+1} \not\mathdeprel{\Gamma_{p+k+1}} b'$ holds for any $\mathindAddone{a'}{b'}{c'} \in \Delta_{p+k+1}$
  if $b_{k} \not\mathdeprel{\Gamma_{p+k}} b$ holds for any $\mathindAddone{a}{b}{c} \in \Delta_{p+k}$.
  We consider cases according to the rule with the conclusion $\Gamma_{p+k}\fCenter\Delta_{p+k}$.
  
  Case 1. 
  If the rule is (\rulename{Weak}), we have the statement by $\Gamma_{p+k+1}\subseteq\Gamma_{p+k}$
  and $\Delta_{p+k+1}\subseteq\Delta_{p+k}$.

  Case 2. 
  If the rule is (\rulename{Subst}),
  we have the statement by Lemma \ref{lemma:subst_rel} \ref{item:lemma-subst_rel-deprel}.
  
  Case 3. 
  If the rule is (\rulenameLa), 
  we have the statement by Lemma \nolinebreak \ref{lemma:eq_rel} \ref{item:lemma-eq_rel-deprel}.
  
  Case 4. The case (\rulename{Case $\mathindAddtwosy$}) with 
  the right assumption $\Gamma_{p+k+1}\fCenter\Delta_{p+k+1}$. 
  Let $\mathindAddtwo{a}{b}{c}$ be the principal formula of the rule.
  There exists $\Pi$ such that 
  $\Gamma_{p+k}\equiv {\mleft(\Pi, \mathindAddtwo{a}{b}{c} \mright)}$ and 
  $\Gamma_{p+k+1}\equiv {\mleft(\Pi, a = sx, b = y, c = z, \mathindAddtwo{x}{sy}{z}\mright)}$
  for fresh variables $x$, $y$, $z$.
  
  We prove this case by contrapositive.
  To show $b_{k} \mathdeprel{\Gamma_{p+k}} b'$,
  assume $b_{k+1} \mathdeprel{\Gamma_{p+k+1}} b'$
  for some $\mathindAddone{a'}{b'}{c'} \in \Delta_{p+k+1}$.
  Define $t$ as $b$ if $b_{k+1}\equiv sy$ and $b_{k+1}$ otherwise.  
  Since $b_{k+1} \mathdeprel{\Gamma_{p+k+1}} b'$ holds, we have $t \mathdeprel{\Gamma_{p+k+1}} b'$.
  By Lemma \ref{lemma:right_asp}, $t \mathdeprel{\Gamma_{p+k}} b'$.
  By $b_{k}\equiv t$, we have $b_{k} \mathdeprel{\Gamma_{p+k}} b'$.

  %%%%%%%%%%%%%%%%%%%%%%%%%%%%%%%%
  Case 5. The case (\rulename{Case $\mathindAddtwosy$}) 
  with the left assumption $\Gamma_{p+k+1}\fCenter\Delta_{p+k+1}$. 
  In this case,  $\Gamma_{p+k}\fCenter\Delta_{p+k}$ is a switching point.
  Let $\mathindAddtwo{a}{b}{c}$ be the principal formula of the rule.
  There exists $\Pi$ such that
  $\Gamma_{p+k} \equiv {\mleft( \Pi, \mathindAddtwo{a}{b}{c}\mright)}$ and
  $\Gamma_{p+k+1} \equiv {\mleft( \Pi, a = 0, b = y, c = y \mright)}$ with a fresh variable $y$. 

  Assume $b_{k} \not\mathdeprel{\Gamma_{p+k}} b''$ for any $\mathindAddone{a''}{b''}{c''} \in \Delta_{p+k}$.
  Fix $\mathindAddone{a'}{b'}{c'} \in \Delta_{p+k+1}$ to show $b_{k+1} \not\mathdeprel{\Gamma_{p+k+1}} b'$.
  By $b_{k+1} \equiv b_{k}$ and $\Delta_{p+k}\equiv\Delta_{p+k+1}$,
  we have $b_{k+1} \not\mathdeprel{\Gamma_{p+k}} b'$. 
  From Lemma \ref{lemma:abc_relations},
  $b' \not\mathdeprel{\Gamma_{p+k}} a$ and $b' \not\mathdeprel{\Gamma_{p+k}} 0$.
  Since $y$ is fresh, we have $b' \not\mathdeprel{\Gamma_{p+k}} y$.
  Since $\Gamma_{p+k}\fCenter\Delta_{p+k}$ is a switching point, $b' \not\mathdeprel{\Gamma_{p+k}} b$.
  By Lemma \ref{lemma:annoying}, $\Gamma_{p+k}\fCenter\Delta_{p+k}$ is an index sequent.
  By Definition \ref{def:annoying} and $\Delta_{p+k}\equiv\Delta_{p+k+1}$,
  $b' \not\mathdeprel{\Gamma_{p+k}} c$.
  By Lemma \ref{lemma:left_asp_1}, $b_{k+1} \not\mathdeprel{\Gamma_{p+k+1}} b'$.
  
  %%%%%%%%%%%%%%%%%%%%%%%%%%%%%%%%%%

  Case 6. 
  The case (\rulename{$\mathindAddonesy$ R${}_2$}).

  In this case,  
  $\Gamma_{p+k}$ is the same as $\Gamma_{p+k+1}$
  and the second argument of a formula with $\mathindAddtwosy$ or $\mathindAddonesy$ 
  in $\Gamma_{p+k} \fCenter \Delta_{p+k}$ is the same as that in $\Gamma_{p+k+1} \fCenter \Delta_{p+k+1}$.
  We thus have the statement. 

  %%%%%%%%%%%%%%%%%%%%%%%%%%%%%%%%%%%%%

  \noindent \ref{item:lemma-index_regress}
  Let $d_{k}=n$.
  
  Case 1. The case (\rulename{Weak}).
  
  If $b_{k+1} \not\mathdeprel{\Gamma_{p+k+1}} b$ for any  $\mathindAddone{a}{b}{c}\in\Delta_{p+k+1}$,
  then $d_{k+1}=\bot$.

  Assume $b_{k+1} \mathdeprel{\Gamma_{p+k+1}} b$ for some $\mathindAddone{a}{b}{c}\in\Delta_{p+k+1}$.
  By Definition \ref{def:deprel}, there exist $m$, $l\in\mathnat$ such that
  $s^{m}b_{k+1}\mathequivrel{\Gamma_{p+k+1}} s^{l} b$.
  By $\Gamma_{p+k+1}\subseteq\Gamma_{p+k}$, we have $s^{m}b_{k+1}\mathequivrel{\Gamma_{p+k}} s^{l} b$.
  Since $b_{k}\equiv b_{k+1}$, we have $s^{m}b_{k}\mathequivrel{\Gamma_{p+k}} s^{l} b$.
  Since $\Delta_{p+k+1}\subseteq \Delta_{p+k}$ holds, we have $\mathindAddone{a}{b}{c}\in\Delta_{p+k}$.
  By $d_{k}=n$, we have $l-m=n$.
  Thus, $d_{k+1}=n$.

  Case 2. The case (\rulename{Subst}) with a substitution $\theta$.
  Note that $b_{k}\equiv b_{k+1}\mleft[\theta\mright]$.
  
  If $b_{k+1} \not\mathdeprel{\Gamma_{p+k+1}} b$ for any $\mathindAddone{a}{b}{c}\in \Delta_{p+k+1}$, 
  then $d_{k+1}=\bot$.
  
  Assume that $b_{k+1} \mathdeprel{\Gamma_{p+k+1}} b$ for some $\mathindAddone{a}{b}{c}\in \Delta_{p+k+1}$.
  By Definition \ref{def:deprel},
  there exist $m$, $l\in\mathnat$ such that $s^{m}b_{k+1}\mathequivrel{\Gamma_{p+k+1}} s^{l} b$.
  By Lemma \ref{lemma:subst_rel} \ref{item:lemma-subst_rel-equiv},
  $s^{m}b_{k+1}\mleft[\theta\mright] \mathequivrel{\Gamma_{p+k}} s^{l} b\mleft[\theta\mright]$.
  Since $b_{k}\equiv b_{k+1}\mleft[\theta\mright]$ holds,
  we have $s^{m}b_{k}\mathequivrel{\Gamma_{p+k}} s^{l}b\mleft[\theta\mright]$.
  Since $\Delta_{p+k}\equiv\Delta_{p+k+1}\mleft[\theta\mright]$ holds,
  we have
  $\mathindAddone{a\mleft[\theta\mright]}{b\mleft[\theta\mright]}{c\mleft[\theta\mright]}\in \Delta_{p+k}$.
  By $d_{k}=n$, we have $l-m=n$.
  Thus, $d_{k+1}=n$.

  %%%%%%%%%%%%%%%%%%%%%%%%%%%%%%

  \noindent \ref{item:lemma-index_not_change}
  Let $d_{k}=n$.

  Case 1. The case (\rulenameLa) with the principal formula $u_1 = u_2$.
  
  Let $b_{k}\equiv b\mleft[v_1:=u_1, v_2:=u_2\mright]$ and $b_{k+1}\equiv b\mleft[v_1:=u_2, v_2:=u_1\mright]$
  for variables $v_1$, $v_2$.
  
  By $d_{k}=n$, there exist $m$, $l\in\mathnat$ such that
  $s^{m}b\mleft[v_1:=u_1, v_2:=u_2\mright] \mathequivrel{\Gamma_{p+k}} s^{l}b\mleft[v_1:=u_1, v_2:=u_2\mright]$
  for some 
  $\mathindAddone{a\mleft[v_1:=u_1, v_2:=u_2\mright]}{b\mleft[v_1:=u_1, v_2:=u_2\mright]}{c\mleft[v_1:=u_1, v_2:=u_2\mright]} \in \Delta_{p+k}$
  and $l-m=n$.
  From Lemma \nolinebreak \ref{lemma:eq_rel} \ref{item:lemma-eq_rel-equiv},
  $s^{m}b\mleft[v_1:=u_2, v_2:=u_1\mright] \mathequivrel{\Gamma_{p+k+1}} s^{l}b\mleft[v_1:=u_2, v_2:=u_1\mright]$.
  Moreover,
  $\mathindAddone{a\mleft[v_1:=u_2, v_2:=u_1\mright]}{b\mleft[v_1:=u_2, v_2:=u_1\mright]}{c\mleft[v_1:=u_2, v_2:=u_1\mright]} \in \Delta_{p+k+1}$.
  Thus, $d_{k+1}=l-m=n$.

  Case 2. The case (\rulename{$\mathindAddonesy$ R${}_2$}).

  Since $\tau_{p+k+1}\equiv\tau_{p+k}$ holds,
  $\Gamma_{p+k}$ is the same as $\Gamma_{p+k+1}$ and 
  the second argument of a formula with $\mathindAddonesy$ in
  $\Delta_{p+k}$ is the same as that in $\Delta_{p+k+1}$, 
  we have $d_{k+1}=d_k$.

  %%%%%%%%%%%%%%%%%%%%%%%%%%%%%%%%%%%%%%%%%%%%%%%

  \noindent \ref{item:lemma-index_Case}
  Let $d_{k}=n$.
  Let $\mathindAddtwo{a}{b}{c}$ be the principal formula of
  the rule (\rulename{Case $\mathindAddtwosy$}) with the conclusion $\Gamma_{p+k}\fCenter\Delta_{p+k}$.

  \noindent \ref{item:lemma-index_left_ass}
  The case where $\Gamma_{p+k+1}\fCenter\Delta_{p+k+1}$ is the left assumption of the rule.
  In this case, $\Gamma_{p+k}\fCenter\Delta_{p+k}$ is a switching point.

  There exists $\Pi$ such that $\Gamma_{p+k} \equiv {\mleft(\Pi, \mathindAddtwo{a}{b}{c}\mright)}$ and 
  $\Gamma_{p+k+1} \equiv {\mleft(\Pi, a=0, b=y, c=y\mright)}$ with a fresh variable $y$.
  By $d_{k}=n$,
  there exist $m$, $l\in\mathnat$ such that $s^{m}b_{k} \mathequivrel{\Gamma_{p+k}} s^{l}b'$
  for some $\mathindAddone{a'}{b'}{c'} \in \Delta_{p+k}$ and $l-m=n$.
  Since the set of formulas with $=$ in $\Gamma_{p+k+1}$ 
  includes the set of formulas with $=$ in $\Gamma_{p+k}$,
  we have $s^{m}b_{k} \mathequivrel{\Gamma_{p+k+1}} s^{l}b'$.
  By $\tau_{k+1}\equiv\tau_{k}$, we have $s^{m}b_{k+1} \mathequivrel{\Gamma_{p+k+1}} s^{l}b'$.
  Since $\Delta_{p+k}\equiv\Delta_{p+k+1}$, we have $\mathindAddone{a'}{b'}{c'} \in \Delta_{p+k+1}$.
  Thus, $d_{k+1}=l-m=n$.

  \noindent \ref{item:lemma-index_not_progress}
  The case where $\Gamma_{p+k+1}\fCenter\Delta_{p+k+1}$ is the right assumption of the rule and
  $\tau_{k}$ is not a progress point of the trace.

  Since $\tau_{k}$ is not a progress point of the trace, we have $\tau_{k+1} \equiv \tau_{k}$.
  By $d_{k}=n$,
  there exist $m$, $l\in\mathnat$ such that $s^{m}b_{k} \mathequivrel{\Gamma_{p+k}} s^{l}b'$
  for some $\mathindAddone{a'}{b'}{c'} \in \Delta_{p+k}$ and $l-m=n$.
  Since the set of formulas with $=$ in $\Gamma_{p+k}$ 
  includes the set of formulas with $=$ in $\Gamma_{p+k+1}$,
  we have $s^{m}b_{k} \mathequivrel{\Gamma_{p+k+1}} s^{l}b'$.
  By $\tau_{k+1}\equiv\tau_{k}$,
  we have $s^{m}b_{k+1} \mathequivrel{\Gamma_{p+k+1}} s^{l}b'$.
  Since $\Delta_{p+k}\equiv\Delta_{p+k+1}$ holds, we have $\mathindAddone{a'}{b'}{c'} \in \Delta_{p+k+1}$.
  Thus, $d_{k+1}=l-m=n$.

  \noindent \ref{item:lemma-index_progress}
  The case where $\Gamma_{p+k+1}\fCenter\Delta_{p+k+1}$ is the right assumption of the rule and
  $\tau_{k}$ is a progress point of the trace.

  There exists $\Pi$ such that 
  $\Gamma_{p+k}\equiv {\mleft(\Pi, \mathindAddtwo{a}{b}{c} \mright)}$ and 
  $\Gamma_{p+k+1}\equiv {\mleft(\Pi, a = sx, b = y, c = z, \mathindAddtwo{x}{sy}{z}\mright)}$
  for fresh variables $x$, $y$, $z$.
  Since $\tau_{k}$ is a progress point of the trace,
  we have $\tau_{k} \equiv \mathindAddtwo{a}{b}{c}$ and
  $\tau_{k+1} \equiv \mathindAddtwo{x}{sy}{z}$.
  Therefore, $b_{k}\equiv b$ and $b_{k+1}\equiv sy$.
  By $d_{k}=n$,  there exist $m$, $l\in\mathnat$ such that
  $s^{m}b \mathequivrel{\Gamma_{p+k}} s^{l}b'$ for some $\mathindAddone{a'}{b'}{c'} \in \Delta_{p+k}$
  and $l-m=n$.
  Since the set of formulas with $=$ in $\Gamma_{p+k+1}$
  includes the set of formulas with $=$ in $\Gamma_{p+k}$,
  we have $s^{m}b \mathequivrel{\Gamma_{p+k+1}} s^{l}b'$.
  By $b\mathequivrel{\Gamma_{p+k+1}} y$,
  we have $s^{m}y \mathequivrel{\Gamma_{p+k+1}} s^{l}b'$.
  Hence, $s^{m}b_{k+1} \mathequivrel{\Gamma_{p+k+1}} s^{l+1}b'$.
  Thus, $d_{k+1}=l+1-m=n+1$.
 \end{proof}
 
 \begin{lem}%
  \label{lemma:key_lemma}
  For an infinite index path $\mleft(\Gamma_{i}\fCenter\Delta_{i}\mright)_{i\geq 0}$
  in $\mathtreeunfolding{\mathproofcf}$,
  there exists $l\in\mathnat$ such that the following conditions hold:
  \begin{enumerate} 
   \item $\Gamma_{l}\fCenter\Delta_{l}$ is a switching point in $\mathtreeunfolding{\mathproofcf}$, and
   \item $\Gamma_{l+1}\fCenter\Delta_{l+1}$ is the right assumption of the rule with the conclusion
	 $\Gamma_{l}\fCenter\Delta_{l}$.
  \end{enumerate}
 \end{lem}
 \begin{proof}
  Since $\mleft(\Gamma_{i}\fCenter\Delta_{i}\mright)_{i\geq 0}$ is an infinite path
  and $\mathtreeunfolding{\mathproofcf}$ satisfies the global trace condition,
  there exists an infinitely progressing trace following a tail of the path.
  Let $\mleft( \tau_{k} \mright)_{k \geq 0}$ be an infinitely progressing trace following 
  $\mleft(\Gamma_{i}\fCenter\Delta_{i}\mright)_{i\geq p}$.
  Let $d_{k}$ be the index of $\tau_{k}$ in $\Gamma_{p+k}\fCenter\Delta_{p+k}$. 

  We show that there exists $l \in \mathnat$ such that $d_{l}=\bot$.
  The set $\mathsetintension{d_{k}}{k\geq 0}$ is finite 
  since the set of sequents in $\mleft(\Gamma_{i}\fCenter\Delta_{i}\mright)_{i\geq 0}$ is finite and
  we have a unique index of an atomic formula with $\mathindAddtwosy$ in $\Gamma_{i}\fCenter\Delta_{i}$. 
  Since $\mleft( \tau_{k} \mright)_{k \geq 0}$ is an infinitely progressing trace 
  following $\mleft(\Gamma_{i}\fCenter\Delta_{i}\mright)_{i\geq p}$, 
  if there does not exist $k' \in \mathnat$ such that $d_{k'}=\bot$, %???
  Lemma \ref{lemma:index} implies that $\mathsetintension{d_{k}}{k\geq 0}$ is infinite.
  Thus, there exists $k' \in \mathnat$ such that $d_{k'}=\bot$. 

  Since $\mleft( \tau_{k} \mright)_{k \geq 0}$ is an infinitely progressing trace following 
  $\mleft(\Gamma_{i}\fCenter\Delta_{i}\mright)_{i\geq p}$, 
  there exists a progress point $\tau_{l}$ with $l>k'$.
  By Lemma \ref{lemma:index}, $d_{l}=\bot$.
  Since $\tau_{k}$ is a progress point, 
  $\Gamma_{p+k}\fCenter\Delta_{p+k}$ is a switching point and  
  $\Gamma_{p+k+1}\fCenter\Delta_{p+k+1}$ is the right assumption of the rule.
 \end{proof}

 \begin{definition}[Rightmost path]
  For a derivation tree $\mathprooffig{D}$ and a node $\sigma$ in $\mathprooffig{D}$,
  we define the \emph{rightmost path} from the node $\sigma$ 
  as the path $\mleft( \sigma_{i} \mright)_{0\leq i< \alpha}$ satisfying the following conditions:
  \begin{enumerate}
   \item The node $\sigma_{0}$ is $\sigma$.
   \item If $\sigma_{i}$ is the conclusion of (\rulename{Case $\mathindAddtwosy$}),
	 the node $\sigma_{i+1}$ is the right assumption of the rule.
   \item If $\sigma_{i}$ is the conclusion of the rules 
	 (\rulename{Weak}), (\rulename{Subst}), (\rulenameLa), or (\rulename{$\mathindAddonesy$ R${}_2$}),
    	 the node $\sigma_{i+1}$ is the assumption of the rule.
  \end{enumerate}
 \end{definition}

 \begin{lem}%
  \label{lemma:rightmost} 
  The rightmost path from an index sequent in $\mathtreeunfolding{\mathproofcf}$ is infinite.
\end{lem}

 \begin{proof}
  By Definition \ref{def:bad-path},
  the rightmost path from an index sequent in $\mathtreeunfolding{\mathproofcf}$ is an index path.
  By Lemma \ref{lemma:annoying}, every sequent on the path is an index sequent.
  By Definition \ref{def:annoying}, (\rulename{$\mathindAddonesy$ R${}_1$}) does not occur in the path.  
  Thus, the path is infinite.
 \end{proof}

 \begin{rem}
  For an infinite path in $\mathtreeunfolding{\mathproofcf}$, 
  the corresponding path in $\mathdertreecf$ has a bud.
 \end{rem}

 \subsection{Proof of main theorem} \label{subsec:proof}
 We prove Theorem \ref{thm:main}.

 \begin{proof}[Proof of Theorem \ref{thm:main}]
  \noindent \ref{item:thm-main-provable}
  The derivation tree given in 
  Figure \ref{fig:counter-ex} is the \CLKID\ proof of 
  $\mathindAddtwo{x}{y}{z} \fCenter \mathindAddone{x}{y}{z}$,
  where ($\dagger$) indicates the pairing of the companion with the bud,
  $\mathdertree{D}_{1}$ is the derivation tree in Figure \ref{fig:counter-ex-D1}
  (some applying rules and some labels of rules are omitted for limited space).
  We use the underlined formulas to denote the infinitely progressing trace 
  for the tails of the infinite path. 
  
  \begin{figure}[thbp]
 \centering
 \begin{prooftree}
  \small
  \AxiomC{}
  \rulenamelabel{$\mathindAddonesy$ R${}_{1}$}
  \UnaryInfC{$\fCenter \mathindAddone{0}{y_1}{y_1}$} \doubleLine
  % \rulenamelabel{$=$ L}
  \UnaryInfC{$\begin{aligned} x&=0,\\ y&=y_1,\\ z&=y_1\end{aligned} 
\fCenter  \mathindAddone{x}{y}{z}$} %(1)
  
  \AxiomC{($\dagger$) $\underline{\mathindAddtwo{x}{y}{z}} \fCenter \mathindAddone{x}{y}{z}$}
  % \rulenamelabel{Subst}
  \UnaryInfC{$\underline{\mathindAddtwo{x_1}{sy_1}{z_1}} \fCenter \mathindAddone{x_1}{sy_1}{z_1}$} %(2)
  
  \AxiomC{$\mathdertree{D}_{1}$}\noLine
  \UnaryInfC{$\mathindAddone{x_1}{sy_1}{z_1} \fCenter \mathindAddone{sx_1}{y_1}{z_1}$} %(2)

  \rulenamelabel{Cut} %(2)
  \BinaryInfC{$\underline{\mathindAddtwo{x_1}{sy_1}{z_1}} \fCenter \mathindAddone{sx_1}{y_1}{z_1}$} \doubleLine
  % \rulenamelabel{$=$ L}
  \UnaryInfC{$\begin{aligned}x&=sx_1,\\ y&=y_1,\\ z&=z_1,\end{aligned}\underline{\mathindAddtwo{x_1}{sy_1}{z_1}} \fCenter \mathindAddone{x}{y}{z}$} %(1)

  \rulenamelabel{Case $\mathindAddtwosy$} %(1)
  \BinaryInfC{($\dagger$) $\underline{\mathindAddtwo{x}{y}{z}} \fCenter \mathindAddone{x}{y}{z}$}
 \end{prooftree}
 \caption{The \CLKID\ proof of $\mathindAddtwo{x}{y}{z} \fCenter \mathindAddone{x}{y}{z}$} 
 \label{fig:counter-ex}
\end{figure}

  \noindent \ref{item:thm-main-not-cut-free-provable}
  We show that there exists a sequence $\mleft(c_{i}\mright)_{i\in\mathnat}$ 
  of switching points in $\mathdertreecf$ which satisfies the following conditions:
  \begin{enumerate}
   \item The height of $c_{i}$ is greater than the height of $c_{i-1}$ in $\mathdertreecf$ for $i>0$.
	 \label{item:ci_height}
   \item For any node $\sigma$ on the path from the root to $c_{i}$ in $\mathdertreecf$ excluding $c_{i}$,
	 $\sigma$ is a switching point if and only if
	 the child of $\sigma$ on the path is
	 the left assumption of the rule (\rulename{Case $\mathindAddtwosy$}).%
	 \label{item:ci_is_only_switching}
  \end{enumerate}

  Now, we construct $\mleft(c_{i}\mright)_{i\in\mathnat}$
  and show \ref{item:ci_height} and \ref{item:ci_is_only_switching} by induction on $i$.

  We consider the case $i=0$.

  The rightmost path in $\mathtreeunfolding{\mathproofcf}$ from the root is an infinite index path 
  since $\mathindAddtwo{x}{y}{z} \fCenter \mathindAddone{x}{y}{z}$ is an index sequent and 
  there exists no node which is the left assumption of (\rulename{Case $\mathindAddtwosy$}) on the path.
  By Lemma \ref{lemma:key_lemma}, there exists a switching point on the path.
  Hence, there exists a switching point on the rightmost path from the root in $\mathdertreecf$.
  Let $c_{0}$ be the switching point of the smallest height among such switching points.
  \ref{item:ci_height} and \ref{item:ci_is_only_switching} follow immediately for $c_{0}$.

  We consider the case $i>0$.

  Let $a$ be the left assumption of the rule with the conclusion $c_{i-1}$.
  Because of \ref{item:ci_is_only_switching},
  the path from the root to $c_{i-1}$ is also an index path.
  Since $c_{i-1}$ is a switching point, the path from the root to $a$ is also an index path.
  By Lemma \ref{lemma:annoying}, $a$ is an index sequent.
  By Lemma \ref{lemma:rightmost},
  the rightmost path from $a$ in $\mathtreeunfolding{\mathproofcf}$ is infinite.
  Therefore, there is a bud on the rightmost path in $\mathdertreecf$ from $a$. Let $b$ be the bud.
  
  Let $\pi_{1}$ be the path from the root to $b$ in $\mathdertreecf$ and
  $\pi_{2}$ be the path from $\mathof{\mathcompanioncf}{b}$ to $b$ in $\mathdertreecf$.
  We define the path $\pi$ in $\mathtreeunfolding{\mathproofcf}$ as $\pi_{1}{\pi_{2}}^{\omega}$.
  Let $\mleft(\sigma_{i}\mright)_{0\leq i}$ be $\pi$.
  Because of \ref{item:ci_is_only_switching}, $\pi$ is an index path.
  By Lemma \ref{lemma:key_lemma}, 
  there is a switching point $\sigma_{l}$ and $\sigma_{l+1}$ is the right assumption of the rule.
  Hence, there is a switching point on $\pi_{1}\pi_{2}$ in $\mathdertreecf$ such that
  its child on the path is the right assumption of the rule.
  Define $c_{i}$ as the switching point of the smallest height among such switching points.

  We show $c_{i}$ satisfies the conditions \ref{item:ci_height} and \ref{item:ci_is_only_switching}.

  \ref{item:ci_height}
  By the definition of $c_{i}$, $c_{i}$ is on the path from the root to $b$.
  By the condition \ref{item:ci_is_only_switching}, $c_{i}$ is not on the path from the root to $c_{i-1}$.
  Hence, the height of $c_{i}$ is greater than that of $c_{i-1}$.

  \ref{item:ci_is_only_switching}
  Let $\sigma$ be a node on the path from the root to $c_{i}$ excluding $c_{i}$.
  We can assume $\sigma$ is on the path from $c_{i-1}$ to $c_{i}$ excluding $c_{i}$
  by the induction hypothesis.

  The ``only if'' part:
  Assume that $\sigma$ is a switching point.
  By the definition of $c_{i}$, we see that $\sigma$ is $c_{i-1}$.
  The child of $c_{i-1}$ on the path from the root to $c_{i}$ is $a$,
  which is the left assumption of the rule.

  The ``if'' part:
  Assume that the child of $\sigma$ on the path is the left assumption of the rule.
  Since there is not the left assumption of a rule on the path from $a$ to $c_{i}$,
  we see that $\sigma$ is $c_{i-1}$.
  Thus, $\sigma$ is a switching point.
  
  We complete the construction and the proof of the properties.

  Because of \ref{item:ci_height},  $c_{0}, c_{1}, \ldots$ are all distinct in $\mathdertreecf$.
  Thus, $\mathsetintension{c_{i}}{i\in\mathnat}$ is infinite.
  This is a contradiction since the set of nodes in $\mathdertreecf$ is finite.
 \end{proof}

 %%%%%%%%%%%%%%%%%%%%%%%%%%%%%%%%%%%%%%%%%%%%%%%%%%%%%%%%%%

  Now, we discuss why we cannot apply the proof technique of a counterexample
  to cut-elimination in cyclic proofs for separation logic given in \cite{Kimura2020}.
  In order to show their counterexample is not provable without a cut rule,
  assuming that there exists a cut-free proof of the counterexample,
  they prove that the rightmost path from the root has no infinitely progressing trace 
  following a tail of the path if the path has a companion.
  On the other hand, the rightmost path from the root in a cut-free \CLKID\ pre-proof of 
  $\mathindAddtwo{x}{y}{z} \fCenter \mathindAddone{x}{y}{z}$ 
  might have a companion and an infinitely progressing trace following a tail of the path
  that might use contraction and weakening on antecedents. 
  For example, the rightmost path from the root in Figure \ref{fig:compare}
  has both a companion and an infinitely progressing trace following a tail of the path.
  Thus, we cannot use their proof technique.

  \begin{figure}[tbhp]
 \begin{prooftree}
  \small
  % \AxiomC{}
  % \rulenamelabel{$\mathindAddonesy$ R${}_{1}$}
  % \UnaryInfC{$\fCenter \mathindAddone{0}{y_1}{y_1}$} \doubleLine
  % % \rulenamelabel{$=$ L}
  % \UnaryInfC{$x=0, y=y_1, z=y_1 \fCenter  \mathindAddone{x}{y}{z}$} %(1)
  \AxiomC{}
  \DeduceC{} %(1)
  
  \AxiomC{}
  \DeduceC{} %(2)
    
  \AxiomC{($\heartsuit$) $\underline{\mathindAddtwo{x_1}{sy_1}{z_1}}, \mathindAddtwo{x}{y}{z} \fCenter \mathindAddone{x}{y}{z}$}
  \UnaryInfC{$\underline{\mathindAddtwo{x_{2}}{sy_{2}}{z_{2}}}, \mathindAddtwo{x}{y}{z}  \fCenter \mathindAddone{x}{y}{z}$}
  \UnaryInfC{$x_{1}=sx_{2}, sy_{1}=y_{2}, z_{1}=z_{2}, \underline{\mathindAddtwo{x_{2}}{sy_{2}}{z_{2}}}, \mathindAddtwo{x}{y}{z}  \fCenter \mathindAddone{x}{y}{z}$} %(2)
  
  \rulenamelabel{Case $\mathindAddtwosy$} %(2)
  \BinaryInfC{($\heartsuit$) $\underline{\mathindAddtwo{x_1}{sy_1}{z_1}}, \mathindAddtwo{x}{y}{z} \fCenter \mathindAddone{x}{y}{z}$}
  % \rulenamelabel{$=$ L}
  \UnaryInfC{$x=sx_1, y=y_1, z=z_1, \mathindAddtwo{x_1}{sy_1}{z_1}, \mathindAddtwo{x}{y}{z}  \fCenter \mathindAddone{x}{y}{z}$} %(1)
  
  \rulenamelabel{Case $\mathindAddtwosy$} %(1)
  \BinaryInfC{$\mathindAddtwo{x}{y}{z} \fCenter \mathindAddone{x}{y}{z}$}
 \end{prooftree}
 \caption{ An example to which the technique in \cite{Kimura2020} cannot be applied} 
 \label{fig:compare}
\end{figure}
%#!latexmk -c -gg -lualatex main.tex
 \section{Conclusion}
 \label{sec:conclusion}
 We have shown that
 $\mathindAddtwo{x}{y}{z} \fCenter \mathindAddone{x}{y}{z}$ is not cut-free provable in \CLKID\
 but provable in \CLKID.
 Consequently, we have shown that (\rulename{Cut}) cannot be eliminated in \CLKID.

 Future work would be
 (1) to restrict principal formulas of (\rulename{Cut}) with the same provability,
 (2) to study whether
 (\rulename{Cut}) can be eliminated in \CLKID\ by restricting the language such as unary predicates,
 (3) to study a subsystem of \LKIDOm\ including \CLKID\  which satisfies the cut-elimination property.
%#!latexmk -c -gg -lualatex main.tex
 \section*{Acknowledgments}
 % We thank the anonymous referees for the reviews.
 We are grateful to Dr.\ Daisuke Kimura for giving us helpful suggestions from the early stage of this work.
 We would also like to thank Dr.\ Taro Sekiyama for giving us valuable comments.

\bibliography{myref}

\end{document}